%% file: main.tex
\documentclass[sigconf]{acmart}
\AtBeginDocument{%
  \providecommand\BibTeX{{%
    Bib\TeX}}}

\newif\ifpublish
\publishtrue

\newif\iflongversion
\longversionfalse

\newif\ifrevision
\revisionfalse

\input{macros.tex}

\def\BibTeX{
  {\rm 
    B\kern-.05em\textsc{i\kern-.025em b}
    \kern-.08em T\kern-.1667em\lower.7ex\hbox{E}\kern-.125emX
  }
}

\copyrightyear{2026}
\acmYear{2026}
\setcopyright{cc}
\setcctype{by}
\acmConference[Middleware '26]{27th International Middleware Conference}{December 14--18, 2026}{Tarragona, Spain}
\acmBooktitle{27th International Middleware Conference (Middleware '26), December 14--18, 2026, Tarragona, Spain}
\acmDOI{10.1145/3801927.3858433}
\acmISBN{979-8-4007-2621-7/2026/12}

\begin{document}

\title{
  Finding \sysname: CFT DAG-based Consensus in the WAN [Research]
}

\author{Rithwik Kerur}
\authornote{Both authors contributed equally to this research.}
\email{rkerur@ucsb.edu}
\affiliation{%
  \institution{UCSB}
  \country{USA}
}

\author{Pasindu Tennage}
\authornotemark[1]
\authornote{This work was done during the affiliation at EPFL.}
\email{pasindu.tennage@gmail.com}
\affiliation{%
  \institution{Digital Asset and EPFL}
  \country{Switzerland}
}

\author{Philipp Jovanovic}
\email{p.jovanovic@ucl.ac.uk}
\affiliation{%
  \institution{Mysten Labs and UCL}
  \country{UK}
}

\author{Dahlia Malkhi}
\email{dahliamalkhi@gmail.com}
\affiliation{%
  \institution{UCSB}
  \country{USA}
}

\author{Alberto Sonnino}
\email{alberto@mystenlabs.com}
\affiliation{%
  \institution{Mysten Labs and UCL}
  \country{UK}
}

\author{Igor Zablotchi}
\email{igor.zablotchi@gmail.com}
\affiliation{%
  \institution{Mysten Labs}
  \country{USA}
}

\begin{abstract}
  \input{sections/abstract.tex}

\end{abstract}


\maketitle

\input{sections/introduction-revision}
\input{sections/model.tex}

\input{sections/dissemination-revision}
\input{sections/consensus.tex}
\input{sections/proofs-revision}
\input{sections/implementation.tex}
\input{sections/evaluation.tex}
\input{sections/related-work-revision}
\input{sections/future-work}

\ifpublish
\else
  \input{acks.tex}
\fi

\clearpage
\bibliographystyle{ACM-Reference-Format}
\bibliography{references-new}

\clearpage
\setcounter{section}{0}
\renewcommand{\thesection}{S}
\input{sections/proofs-supplementary}

\end{document}

%% file: macros.tex
\usepackage{graphicx}
\usepackage{textcomp}
\usepackage{xspace}
\usepackage{tikz}
\usepackage{subcaption}
\usepackage{algorithm}
\usepackage{multicol}
\usepackage{listings}
\usepackage[noend]{algpseudocode}
\usepackage{algorithmicx}
\usepackage[capitalize]{cleveref}
\usepackage{enumitem}

\ifpublish
    \newcommand{\missing}[1]{}
    \newcommand{\alberto}[1]{}
    \newcommand{\pasindu}[1]{}
    \newcommand{\philipp}[1]{}
    \newcommand{\DM}[1]{}
    \newcommand{\Rithwik}[1]{}
    \newcommand{\igor}[1]{}
\else
    \newcommand{\missing}[1]{\textcolor{red}{#1}}
    \newcommand{\alberto}[1]{\textcolor{orange}{\textbf{Alberto:} #1}}
    \newcommand{\igor}[1]{\textcolor{blue}{\textbf{Igor:} #1}}
    \newcommand{\pasindu}[1]{\textcolor{green}{\textbf{Pasindu:} #1}}
    \newcommand{\philipp}[1]{\textcolor{purple}{\textbf{Philipp:} #1}}
    \newcommand{\DM}[1]{\textcolor{magenta}{\textbf{Dahlia:} #1}}
    \newcommand{\Rithwik}[1]{{\color{red} Rithwik: #1}}
\fi

\newcommand{\rev}[1]{\ifrevision\textcolor{blue}{#1}\else#1\fi}
\newenvironment{revblock}{\ifrevision\color{blue}\fi}{}
\newcommand{\delblock}[1]{%
  \ifrevision
    \par\noindent
    {\color{red}\small\textbf{[Removed:]}\ #1\par}%
  \fi
}

\newcommand{\pparagraph}[1]{\vskip 0.5em \noindent \emph{#1.}}
\newcommand{\para}[1]{\pparagraph{#1}}

\hypersetup{colorlinks=true,linkcolor=Maroon,citecolor=blue,urlcolor=blue!60!black}

\newcommand{\algsize}{\scriptsize}

\newcommand{\eg}{\textit{e.g.,} }
\newcommand{\ie}{\textit{i.e.,} }

\newtheorem{observation}{Observation}

\newtheorem{thm}{Theorem}
\newtheorem{lemma}{Lemma}

\theoremstyle{definition}

\newcommand{\sysname}{\textsc{Nemo-Nemo}\xspace}
\newcommand{\mysticeti}{Mysticeti\xspace}

\newcommand{\store}{\ensuremath{DAG}\xspace}

\newcommand{\propose}{\textsf{Propose}\xspace}

\newcommand{\decide}{\textsf{Decide}\xspace}

\newcommand{\scommit}{\texttt{commit}\xspace}
\newcommand{\sskip}{\texttt{skip}\xspace}
\newcommand{\sundecided}{\texttt{undecided}\xspace}

%% file: sections/abstract.tex
This paper introduces Nemo-Nemo, a practical crash-fault tolerant (CFT) consensus protocol designed to outperform existing protocols in wide-area networks by bridging design principles from the CFT and Byzantine-fault tolerant (BFT) worlds.
By structuring command propagation through a causally ordered DAG, Nemo-Nemo allows all consensus replicas to propose commands with a naturally self-regulating communication regime. 
By exploiting multi-leader architecture, Nemo-Nemo avoids the performance bottleneck inherent to single-leader protocols where the leader stalls or crashes. 
By separating command dissemination from consensus logic, Nemo-Nemo handles challenging network conditions even when consensus commits are stalled.  
Moreover, leader proposals that miss a deadline are never dropped, but deterministically deferred and executed later, preserving throughput under transient network delays.
And by enabling Nemo-Nemo to commit on a DAG in just two network hops, it matches the latency of existing CFT systems, while achieving significantly higher throughput, and remains resilient under network instability where existing CFT protocols degrade.
The result is a robust, deployable system: the first DAG-based CFT consensus protocol proven to exceed state-of-the-art wide-area network performance in both speed and resilience.

%% file: sections/introduction-revision.tex
\section{Introduction}\label{sec:introduction}

Byzantine fault-tolerant (BFT) consensus systems have seen significant advances, driven by a multi-trillion dollar blockchain industry.
These advances have yielded powerful techniques, notably DAG-based consensus architectures, that enable high throughput, low latency, and graceful handling of both failures and network instabilities (e.g.,~\cite{narwhal-tusk, bullshark, cordial-miners, mysticeti2023} and others). Meanwhile, crash fault-tolerant (CFT) consensus remains the workhorse for traditional distributed systems, yet has not benefited equally from these innovations.

This paper bridges that gap. We argue that CFT systems designed to operate in the WAN can and should adopt lessons from modern BFT consensus, while adapting them to exploit the simpler model of crash-only failures. Specifically:

\begin{enumerate}

    \item \textbf{DAG-based architecture:} We can boost performance and alleviate network and failure-induced hiccups by adopting the DAG-based approach, which provides a self-regulating mempool and seamless failure recovery.

    \item \textbf{Battle-tested implementations:} Rather than building from scratch, we can adapt production-grade BFT codebases to CFT settings.

    \item \textbf{Unstable network testing:} Even without Byzantine replicas, we must test CFT systems under unstable network conditions, such as shifting connected majorities and leaders, which partial synchrony alone does not capture.

\end{enumerate}

This paper presents \sysname, the first DAG-based CFT protocol. \sysname applies these principles to achieve the best performance among known CFT protocols in wide-area network settings under both partial synchrony\cite{dwork1988consensus} and random asynchronous network\cite{danezis2025byzantine} models. 

\para{Vision and motivation}
Modern DAG-based BFT consensus systems offer two central lessons that form the foundation of \sysname. The first is the value of a \textit{causally ordered, and self-regulating mempool}
A mempool functions as a transport layer for disseminating transaction requests, allowing replicas to inject transactions in parallel. In a DAG-based mempool, transactions are bundled into blocks that maintain causal order. This design achieves three goals: (1) it induces a \textit{self-regulating, round-based} communication pattern where replicas wait for a threshold of blocks before advancing rounds; (2) each block automatically carries its causal history, ensuring consistent DAG views across correct replicas; and (3) progress continues seamlessly despite leader failures, enabling rapid recovery when new leaders are promoted.

In principle, any consensus protocol can sit atop such a DAG mempool, and many BFT protocols leverage this idea: Narwhal-HS~\cite{narwhal-tusk} integrates HotStuff \cite{hotstuff} with a DAG and achieves roughly a $50 \times$ performance gain over bare HotStuff, while Autobahn~\cite{giridharan2024autobahn} integrates PBFT \cite{castro99practical} with a DAG and similarly attains significant throughput improvements.

Beyond a DAG-based mempool alone, the second lesson is that consensus protocols can be embedded directly into the DAG structure by pre-designating a \textit{skeleton} of leader proposals and treating blocks that reference those proposals as protocol votes~\cite{baird2020hashgraph}. In this view, much of the logic of traditional consensus protocols collapses into simple structural rules on the DAG. For example, simple and direct commit logic can be realized by $f+1$ blocks referencing a leader proposal~\cite{fino,bullshark-partial-sync}.

\para{Overview of \sysname}
For the DAG-mempool, a key design choice concerns whether DAG blocks must be explicitly certified for availability by a quorum before insertion\cite{Nayak-DAG}. Several BFT DAGs, such as Tusk~\cite{narwhal-tusk} and Bullshark~\cite{bullshark}, use certified blocks, motivated by the need for non-equivocation and data availability in adversarial environments. However, pre-certifying each block adds communication rounds on top of those required for consensus, increasing latency. Alternative BFT DAGs, including Cordial Miners~\cite{cordial-miners} and Mysticeti~\cite{mysticeti2023}, avoid pre-certification to reduce latency.

Prioritizing low latency, \sysname uses non-certified blocks. Note that in CFT settings, equivocation is not a concern, hence the benefit of pre-certification is diminished whereas admitting non-certified blocks enables extremely low-latency.

The DAG \textit{direct} commit rule of \sysname is straightforward: a skeleton block becomes \textit{directly} committed if referenced by $f+1$ blocks in the next round (see \Cref{fig:rounds}).
The crux is maintaining safety against failures across replicas, each distinctly evolving the DAG due to network scheduling and omissions. When a skeleton slot is not directly committed, it can become \textit{indirectly} committed by a future skeleton slot using the (causally ordered) DAG structure itself (see \Cref{fig:example}). Notably, there is no explicit view-change by which a future leader may reinstate the proposed block. Rather, a DAG rule, which is somewhat subtle, is employed to indirectly commit or skip the undecided slot; it is detailed in the body of the paper.

\rev{\para{Pipelining rotating leaders}
Rotating the leader every round distributes load uniformly, but rotation is not free: a two-phase commit rule spans two rounds, so if waves were laid end to end a fresh proposal could enter only every other round. \sysname avoids this by \textit{overlapping} waves. Every round is at once the propose round of a new wave and the decide round of the preceding one (see \Cref{fig:rounds}, right). Each skeleton block thus runs the first phase for its own proposal while its $f+1$ references to the previous round complete the second phase for the preceding one. This is the pipelining HotStuff~\cite{hotstuff} introduced for a chained, stable leader, except that a \textit{different} replica leads every round and the second phase rides on ordinary DAG references, so rotation needs no view change and adds no message to the critical path. A proposal slot therefore opens every round rather than every wave, halving the expected wait before a command can be proposed; with $k$ skeleton slots per round, $k$ proposals enter the pipeline at once (see \Cref{fig:example}).}

\sysname is designed to remain robust under unstable conditions such as transient
network delays. Leader-based consensus protocols face a well-known timeout tradeoff: timeouts that are too short risk discarding correct leader proposals, while timeouts that are too long can stall the system upon leader failure or hang-up. \sysname eliminates this tradeoff by
construction without compromising fairness or throughput. Leader proposals are never dropped; instead, they are deferred, typically to the immediately succeeding round, allowing the system to progress smoothly despite temporary leader disruptions.

\para{Implementing \sysname}
To validate our vision,
we wanted to employ mature code and optimize it for low latency.
We started with Mysticeti, a BFT protocol deployed in production by multiple blockchains~\cite{sui-code,iota-mysticeti,ika-mysticeti}.

Applying Mysticeti to CFT involved both obvious modifications and structural opportunities that fundamentally improve performance, including the following:
(i) changing quorum thresholds from $2f+1$ to $f+1$, (ii) eliminating an entire communication round and simplifying commit rules
to enable Nemo-Nemo to commit on a DAG in just two network hops, like existing CFT systems,
 (iii) simplifying authentication by removing signatures and integrity checks, thus streamlining command ingestion without Byzantine peer concerns.

\para{Unstable Network Testing}
Finally, in addition to evaluating performance under standard, partial synchrony conditions, we address an existing gap in the evaluation methodology for CFT consensus systems.
%
Existing CFT systems are designed to cope with a minority of stragglers, progressing rather well with a stable, well-connected majority. However, they suffer a performance degradation when network conditions are unstable and the fast majority is constantly ``shifting''.
To address this, in addition to traditional evaluations, we developed a framework to evaluate protocols under a recently proposed \textit{random asynchronous model}, which assumes message delivery delays are random but not adversarially controlled. This model captures realistic network variability—where packets may be delayed or reordered—without assuming malicious nodes. Importantly, this model reveals performance characteristics invisible under standard partial synchrony assumptions, where protocols are typically evaluated with stable, low-latency networks. We detail how we simulate this network model in \Cref{sec:implementation}


Our evaluations demonstrate that existing CFT protocols behave quite differently under randomized network delays compared to stable networks. \sysname, by contrast, maintains robust performance across both settings. This validates our motivating insight: by decoupling block propagation from consensus, \sysname remains robust under the random asynchronous model as the DAG continues to grow even in the absence of commits. It contrasts with traditional CFT protocols, which couple these processes, forcing block production to halt until consensus is achieved.

\para{Contributions}
We summarize our contributions as follows:

\begin{enumerate}

    \item \rev{\textbf{Pipelined rotating leadership for CFT.} To our knowledge, \sysname is the first CFT protocol to combine rotating leaders with HotStuff-style pipelining: each leader runs the first phase for a new proposal while completing the second phase for the previous one. This offsets the cost of rotation and halves consensus queuing delay, while still committing in two network hops and requiring no view-change protocol.}

    \item \rev{\textbf{Pipelining with multiple proposers on a DAG.} \sysname integrates this pipeline with multiple proposers per round in a DAG, yielding the throughput of DAG-based consensus together with predictable latency and graceful degradation under faulty leaders. Proposals that miss their deadline are deferred rather than dropped, so a crash causes no visible throughput interruption. \sysname sustains at least $2\times$ the saturation throughput of state-of-the-art CFT protocols at lower median latency.}

    \item \rev{\textbf{An evaluation framework for random asynchrony.} \sysname introduces the first evaluation framework for consensus under the random-asynchrony model~\cite{danezis2025byzantine}, exposing failure modes invisible under conventional partial-synchrony benchmarks: Multi-Paxos and QuePaxa lose $25\%$ and $20\%$ of their saturation throughput, while \sysname's is unaffected, outperforming them by $3\times$ and $2.25\times$.}

\end{enumerate}

\delblock{We present \sysname, the first DAG-based CFT consensus protocol. \sysname achieves the best performance among known CFT protocols in WAN settings under both partial synchrony and random asynchronous network models.
Under the random asynchronous network model, our results show that existing CFT protocols degrade significantly under this model, while \sysname maintains robust performance.
We show how, unlike traditional CFT protocols that require explicit view-change mechanisms, \sysname integrates leader rotation and multi-proposer rounds directly into the DAG structure, eliminating protocol complexity and latency overhead.
We exploit CFT's lack of equivocation to admit non-certified blocks into the DAG, achieving commit latencies unattainable in BFT settings. This design choice, along with optimized commit rules and block ingestion, enables \sysname to commit in just two network hops, like existing CFT systems, while achieving a significantly higher saturation throughput of at least 2x.
(Superseded by the three-fold list above; the uncertified-DAG and two-hop claims are folded into items 1 and 2.)}

%% file: sections/model.tex
\section{System Model and Assumptions} \label{sec:model}

\sysname is a peer-to-peer message-passing system which implements atomic broadcast~\cite{cachin2011introduction}, with the following properties:
(1) \textit{Validity}: If a correct replica $p$ proposes a value $v$, then $p$ eventually
commits $v$; (2) \textit{No duplication}: No value is committed more than once; (3) \textit{No creation}: If a replica commits a value $v$ with proposer $s$, then $v$ was previously proposed by replica $s$; (4) \textit{Agreement}: If a value $v$ is committed by some correct replica, then $v$ is eventually committed by every correct replica; (5) \textit{Total order}: Let $v_1$ and $v_2$ be any two values and suppose $p$ and $q$ are any two correct replicas that commit $v_1$ and $v_2$. If $p$ commits $v_1$ before $v_2$, then $q$ commits $v_1$ before $v_2$.
In the random asynchronous model, validity and agreement are satisfied with probability $1$ (almost surely).

\para{Fault model}
\sysname follows the standard CFT assumptions used in prior work~\cite{lamport2001paxos, ongaro2014search}. The system consists of $n = 2f + 1$ replicas, of which at most $f$ may crash. A replica is \emph{correct} if it follows the protocol and \emph{crashed} otherwise.

\para{Network model}
We analyze and evaluate \sysname under two network models: the standard partial synchrony model and the more recent random asynchronous model. Links between correct parties behave reliably, and messages among them eventually arrive.

In the partially synchronous model, we assume the standard Global Stabilization Time (GST) used throughout related work. After GST, every message sent among correct parties arrives within $\Delta$ time. Before GST, the protocol operates in a fully asynchronous network where messages eventually arrive but no bound limits their delay.

Although \sysname operates in the CFT setting, it is designed to operate in a more unstable environment prone to network instabilities. To model this, we employ the random asynchronous model~\cite{danezis2025byzantine} which assumes that message delivery follows a random schedule. The random asynchronous model strikes a midpoint between the partially synchronous and fully asynchronous settings: it does not rely on synchronous periods to commit, but relaxes the worst-case adversarial scheduling of the classic asynchronous model~\cite{fischer1985,ben1983another}, where an adversary can control the message schedule indefinitely.


%% file: sections/dissemination-revision.tex
\section{The \sysname DAG: Data Dissemination}\label{sec:dissemination}

\sysname operates in two logical layers: (i) a data dissemination layer that propagates client transactions through a structured DAG (this section); and (ii) an ordering layer interpreting the DAG to derive a total order of transactions (see \Cref{sec:consensus}).

\subsection{Block creation}
Replicas operate in a sequence of logical \emph{rounds}, and in each round, every replica proposes a unique \emph{block}.
During a round, replicas receive transactions from clients and blocks from other replicas. Clients submit transactions to a replica, which includes them in its (next) block. If a transaction does not finalize quickly enough, the client simply resubmits it to a different replica. Each block references blocks from prior rounds, starting with the author's most recent block, and includes \emph{fresh transactions} not yet included in preceding blocks. 
Specifically, a block contains at least the following elements: (1) the author $A$ of the block; (2) a round number $R$; (3) a list of transactions; and (4) at least $f+1$ distinct references to blocks from the previous round $R-1$, as well as any additional references from earlier rounds, for which the replica has already downloaded the complete causal history.
\rev{\Cref{alg:main} specifies this layer: \textsc{CreateBlock} assembles the elements above, and \textsc{OnBlockReceived} admits a block into the local DAG only once its causal history is complete.}
At the core of \sysname is a round-based DAG structure where each vertex is a block. Each replica maintains a local view of the DAG, adding a block once it has downloaded its entire causal history. Consequently, all correct replicas eventually converge on the same DAG view.

\subsection{Inclusive block proposing}\label{sec:block-proposing}

Unlike traditional leader-based consensus protocols, where only leaders propose values, in \sysname all replicas, leaders or not, propose values in each round. In each round, a subset of replicas known to everyone is pre-designated as \textit{skeleton} nodes.
Skeleton nodes are akin to leaders (possibly multiple ones per view, as in multi-leader protocols). That is, blocks by skeleton nodes drive progress in the protocol: they can become directly committed, and they determine the commit decisions for blocks they causally reference, as described in the ordering layer (see \Cref{sec:consensus}).
A replica disseminates its block only after receiving (in addition to $f+1$ blocks from the previous round, also) all skeleton blocks from the previous round, or after a timeout expires\rev{ (\textsc{TryPropose}, \Cref{alg:main})}.
This ensures liveness once the system becomes synchronous \rev{(see \Cref{sec:liveness-proofs})}.

\para{Random Asynchronous Model}
This model assumes that the first $f+1$ blocks received by a replica in a round $R$ constitute a random sample. \rev{Concretely, we assume message transmission delays are independent and identically distributed across links and rounds, so that the order in which a replica receives a round's blocks is uniformly random.} We discuss the implementation of this mechanism in \Cref{sec:implementation} and prove that \sysname is live with probability 1 under this model \rev{(see Supplementary Material ~\cite{supp})}.


%% file: sections/consensus.tex
\section{The \sysname Consensus Protocol} \label{sec:consensus}

By itself, the DAG layer described in \Cref{sec:dissemination} functions as a scalable data dissemination layer. It grows with the network and ensures that all disseminated transactions are reliably available. However, a consensus layer is required to order transactions.
\sysname integrates this layer logically into the DAG messages, rather than sending separate consensus messages, and infers decisions directly from the DAG itself.

\subsection{Identifying DAG patterns}

\sysname leverages patterns in the DAG structure to define its decision rules.

\para{Rounds and waves}
\Cref{fig:rounds} (left) shows an example of a \sysname DAG with three replicas, $(N_0, N_1, N_2)$. \sysname defines a \emph{wave} of two rounds for every block. The first round $R$ (\propose) includes the blocks that the wave attempts to commit ($P_0$, $P_1$, $P_2$). The second round $R+1$ (\decide) includes the blocks ($D_0$, $D_1$, $D_2$) implicitly acting as ``votes''~\cite{lamport2001paxos} determining the blocks of the previous round to commit.
\Cref{fig:rounds} (right) shows that \sysname initiates a new wave every round, with two consecutive waves always overlapping by one round.
Put differently, each round starts a new wave.
\Cref{alg:decider} (\Cref{sec:decision-rule}) formally defines a wave.

\begin{figure}[t]
    \centering
    \vskip -1em
    \includegraphics[width=\columnwidth]{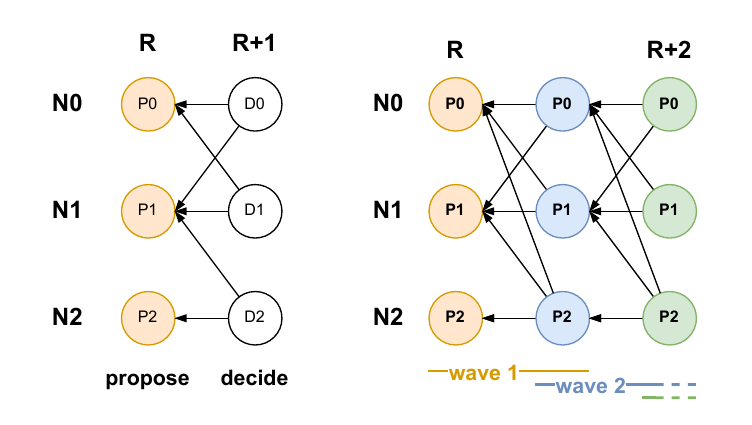}
    \vskip -0.5em
    \caption{
        The structure of the \sysname DAG. Left: The structure of a wave, consisting of two rounds (\propose and \decide). Right: Waves patterns in the \sysname protocol (each round starts a new overlapping wave).
     }
    \label{fig:rounds}
\end{figure}

\para{Skeleton slots}
A skeleton slot (or simply ``a slot'') is a tuple (replica, round) and can be either empty or contain the replica's proposal for the respective round. The slot 
can assume one of three states: \scommit, \sskip, or \sundecided. All slots are initially set to \sundecided and the goal of the protocol is to classify them as \scommit or \sskip.

The number of skeleton slots instantiated per round can be configured, and for systems with few faults, it can be as high as $n$ so that every block has a chance to commit in two steps. It can also be dynamically adjusted based on the network conditions, following a similar deterministic approach to HammerHead~\cite{hammerhead}. The protocol is initialized with a deterministic total order among skeleton slots, known to all replicas. Within a single round, this ordering may reflect a fixed (\eg round robin) or a variable regime per round. \Cref{fig:example} illustrates an example of a \sysname DAG with three replicas, $(N_0, N_1, N_2)$, two skeleton slots per round, and a potential skeleton slot ordering represented as ($S_{0a}$, $S_{0b}$) and ($S_{1a}$, $S_{1b}$) for the first and second rounds, respectively.

\subsection{The \sysname decision rule}\label{sec:decision-rule}

We present the decision rule of \sysname leveraging an example protocol run depicted in \Cref{fig:example}. In this example, we refer to blocks using the notation $B_{(N_i, R)}$, where $N_i$ is the issuing replica and $R$ is the block's round.

\begin{figure}[t]
    \centering
    \vskip -1em
    \includegraphics[width=\columnwidth]{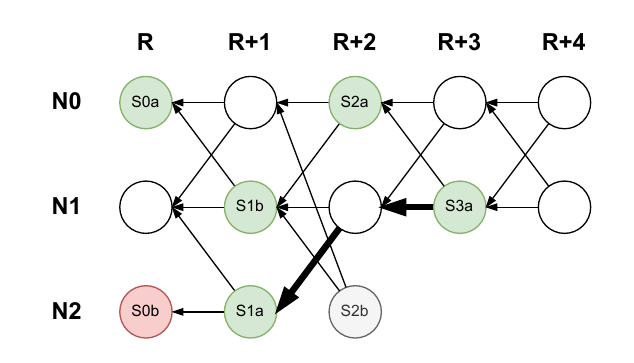}
    \vskip -0.5em
    \caption{
        Example execution with three replicas and two skeleton slots per round. Skeleton blocks are classified as \scommit (green), \sskip (red), or \sundecided (grey). Thick arrows signify commits via the indirect decision rule.
    }
    \label{fig:example}
\end{figure}

All skeleton slots are initially in the \sundecided state. The replica holds the portion of the DAG depicted in \Cref{fig:example} and attempts to classify as many blocks in the skeleton slots as possible as either \scommit or \sskip.

\para{Step 1: Direct decision rule}
To start the classification process, the replica processes each slot individually, starting with the highest ($S_{3a}$), applying the \sysname \emph{direct decision rule}. The replica classifies a block $B$ in a slot as \scommit if it observes $f+1$ blocks (\ie two blocks in our example) from the subsequent round referencing it. Otherwise, the replica leaves the slot as \sundecided (for now). This rule is formally described by the function \textsc{TryDirectDecide} in \Cref{alg:decider}.

In \Cref{fig:example}, the replica targets $S_{3a}$ first. It observes that $B_{(N_0, R+4)}$ and $B_{(N_1, R+4)}$ reference it. Therefore, it classifies $S_{3a}$ as \scommit. \Cref{sec:evaluation} shows that this scenario is the most common (in the absence of an asynchronous adversary) and results in the lowest latency.
Blocks $S_{2a}$, $S_{1b}$, and $S_{0a}$ are also classified as \scommit as they each have $f+1$ references from the subsequent round.

\para{Step 2: Indirect decision rule}
In the case where the direct decision rule cannot classify a skeleton slot  $B_{(N_i, R)}$, the replica uses the \sysname \emph{indirect decision rule}. This rule looks at future slots to decide about the current one. First, it finds an \emph{anchor}. This is the earliest non-skipped skeleton slot with a round number $R' > R + 1$ that is either still classified as \sundecided or already classified as \scommit. Recall that slots are totally-ordered, so the anchor is uniquely defined regardless of how the DAG evolves at different replicas. If the anchor is \sundecided, the replica marks the current slot as \sundecided. If the anchor is \scommit, the replica checks if it indirectly references the target slot, that is, it checks whether there is a path between the anchor and the target skeleton slot. If it does, the replica marks the target slot as \scommit. If it does not, the replica marks it as \sskip.
This rule is formally described by the function \textsc{TryIndirectDecide} in \Cref{alg:decider}.
\Cref{sec:proofs} shows that the direct and indirect decision rules are consistent: if one replica directly commits a block, no honest replicas will indirectly skip it (and vice versa).

In this example, the replica fails to classify $S_{2b}$, $S_{1a}$, and $S_{0b}$ using the direct decision rule and thus searches for their respective anchors. The status of the anchor of $S_{2b}$ is still \sundecided; the replica thus marks $S_{2b}$ as \sundecided. Eventually, the DAG will grow and a block with round $R' > R+3$ will become the anchor for $S_{2b}$. For the moment however, $S_{2b}$ remains \sundecided. $S_{3a}$ is the anchor for $S_{1a}$; since $S_{3a}$ has a path to $S_{1a}$ (marked with thick arrows), the replica classifies $S_{1a}$ as \scommit. The anchor of $S_{0b}$ is $S_{2a}$ which does not have a path to $S_{0b}$. Thus, the replica classifies $S_{0b}$ as \sskip.

\para{Step 3: Skeleton slots sequence}
After processing all slots, the replica derives an ordered sequence of the blocks contained in the skeleton slots. It then iterates over this sequence, committing all slots marked as \scommit and skipping all slots marked as \sskip. This process continues until the replica encounters the first \sundecided slot. As shown in \Cref{sec:proofs} this commit sequence is safe, and eventually, all slots will be classified as either \scommit or \sskip.
In the example shown in \Cref{fig:example}, the skeleton block output by the replica is [$S_{0a}$, $S_{1a}$, $S_{1b}$, $S_{2a}$]. Since $S_{2b}$ remains \sundecided, it is not included in the sequence and $S_{3a}$ is also excluded.

\para{Step 4: Commit sequence}
Following the approach introduced by DagRider~\cite{dag-rider}, the replica determines a linear ordering of all the blocks within the sub-DAG defined by each skeleton block, including previously skipped skeleton blocks (if they are reachable), by performing a depth-first search. If a previous skeleton slot has already linearized a block, it is not re-linearized. The replica processes skeleton slots sequentially, ensuring that all blocks are included in the final commit sequence in the correct order, according to their causal dependencies. The procedure $\textsc{LinearizeSubDags}$ of \Cref{alg:helper} formally describes this linearization process.

In this example, $S_{0a}$ and $S_{0b}$ do not define any sub-DAG (because the example begins at round $R$) and are thus directly added to the commit sequence. Next, $S_{1a}$ defines the sub-DAG $\{ L_{1a}, B_{(N_1, R)}, L_{0b} \}$, which is linearized as [$B_{(N_1, R)}$, $S_{0b}$, $S_{1a}$]. The replica continues this process for each skeleton, linearizing the sub-DAGs defined by $S_{1b}$ as [$S_{1b}$], since both $S_{0a}$ and $B_{(N_1, R)}$ are already part of the commit sequence, and so forth. The final commit sequence is
    [
        $S_{0a}$,
        $B_{(N_1,R)}$, $S_{1a}$,
        $S_{1b}$,
        $B_{(N_0,R+1)}$, $S_{2a}$
    ].

\input{algorithms/core}
\input{algorithms/decider}
\input{algorithms/helper}

%% file: algorithms/core.tex
\begin{algorithm}[t]
    \algsize
    \caption{\sysname (at replica $A$)}
    \label{alg:main}
    \begin{algorithmic}[1]

        \State \texttt{leadersPerRound} \Comment{A number between 1 and $2f+1$}
        \State \texttt{waveLength} $= 2$
        \State $r \gets 1$; \ $DAG[\ast] \gets \emptyset$ \Comment{Current round; local DAG view, by round}
        \State $pending \gets [ \; ]$ \Comment{Client transactions, appended on arrival}
        \State $included \gets \emptyset$; \ $timeout$ \Comment{Own proposed txs; round timer, re-invokes \Call{TryPropose}{\,} on expiry}
        \Statex

        \Procedure{OnBlockReceived}{$b$}
        \If{$\exists b' \in b.parents : b' \notin \bigcup_{r' \geq 1} DAG[r']$}
        \State \textit{Request missing ancestors, defer $b$ and \Return} \Comment{Retried when complete}
        \EndIf
        \State $DAG[b.round] \gets DAG[b.round] \cup \{b\}$; \ \Call{TryPropose}{\,}
        \EndProcedure
        \Statex
        
        \Procedure{TryPropose}{\,} \Comment{Attempt to propose at round $r$}
        \If{$|DAG[r{-}1]| < f{+}1 \ \vee\ (\neg\textsc{Expired}(timeout) \ \wedge\ \exists\, l \in [0, \texttt{leadersPerRound}) : \textsc{GetLeaderBlock}(\textsc{WaveNumber}(r{-}1),\, l) = \bot)$}
        \State \Return \Comment{Too few blocks, or awaiting skeletons}
        \EndIf
        \State $b \gets$ \Call{CreateBlock}{$r$}
        \State \Call{Broadcast}{$b$}; \ \Call{OnBlockReceived}{$b$} \Comment{Best-effort; no certification}
        \State $r \gets r + 1$; \ \Call{Reset}{$timeout$}
        \EndProcedure
        \Statex
        
        \Procedure{CreateBlock}{$r'$}
            \State $P \gets DAG[r'{-}1]$ \Comment{Can also ref older blocks not reachable from $P$}
            \State $T \gets \{t \in pending : t \notin included\}$; \ $included \gets included \cup T$
            \State \Return $Block(A,\ r',\ P,\ T)$
        \EndProcedure
        \Statex

        \Procedure{TryDecide}{$r_{committed}, r_{highest}$}
        \State $S \gets [ \; ]$ \Comment{Holds decisions}
        \For{$r \gets r_{highest}$ \textbf{down to} $r_{committed} + 1$}
        \For{$l \gets \texttt{leadersPerRound} - 1$ \textbf{down to} $0$}
        \State $i \gets r \; \bmod$ \texttt{waveLength}
        \State $D \gets$ \Call{Decider}{$i, l$}
        \State $w \gets D.$\Call{WaveNumber}{$r$}
        \State $s \gets D.$\Call{TryDirectDecide}{$w$}
        \If{$s = \bot$} $s \gets D.$\Call{TryIndirectDecide}{$w, S$}
        \EndIf
        \State $S \gets s \parallel S$
        \EndFor
        \EndFor
        \State \Return $S$
        \EndProcedure
        \Statex

        \Procedure{ExtendCommitSequence}{$r_{committed}, r_{highest}$}
        \State $S \gets$ \Call{TryDecide}{$r_{committed}, r_{highest}$}
        \State $S_{commit} \gets [ \; ]$ \Comment{Holds committed blocks}
        \For{$s \in S$}
        \If{$s = \bot$} \textbf{break}
        \EndIf
        \If{$s = \texttt{Commit}(b_{leader})$} $S_{commit} \gets S_{commit} \parallel b_{leader}$
        \EndIf
        \EndFor
        \State \Return \Call{LinearizeSubDags}{$S_{commit}$} \Comment{Commit returned blocks}
        \EndProcedure

    \end{algorithmic}
\end{algorithm}

%% file: algorithms/decider.tex
\begin{algorithm}[t]
    \algsize
    \caption{Decider Instance}
    \label{alg:decider}
    \begin{algorithmic}[1]

        \State \texttt{waveOffset} $= i$ \Comment{The first parameter of the Decider (i)}
        \State \texttt{leaderOffset} $= l$ \Comment{The second parameter of the Decider (l)}
        \State \texttt{waveLength} $= 2$
        \Statex

        \Procedure{WaveNumber}{$r$}
        \State \Return $(r - \texttt{waveOffset}) / \texttt{waveLength}$
        \EndProcedure
        \Statex

        \Procedure{ProposeRound}{$w$}
        \State \Return $(w \cdot \texttt{waveLength}) + \texttt{waveOffset}$
        \EndProcedure
        \Statex

        \Procedure{DecisionRound}{$w$}
        \State \Return \Call{ProposeRound}{$w$}$ + (\texttt{waveLength} - 1)$
        \EndProcedure
        \Statex

        \Procedure{SupportedLeader}{$w, b_{leader}$}
        \State $B_{decision} \gets$ \Call{GetDecisionBlocks}{$w$}
        \State \Return $|\{ b' \in B_{decision} : $ \Call{IsLink}{$b', b_{leader}$}$ \}| \geq f + 1$
        \EndProcedure
        \Statex

        \Procedure{TryDirectDecide}{$w$}
        \State $b_{leader} \gets$ \Call{GetLeaderBlock}{$w$, \texttt{leaderOffset}}
        \If{\Call{SupportedLeader}{$w, b_{leader}$}} \Return \texttt{Commit}$(b_{leader})$
        \Else \; \Return $\bot$
        \EndIf
        \EndProcedure
        \Statex

        \Procedure{TryIndirectDecide}{$w, S$}
        \State $r_{decision} \gets $\Call{DecisionRound}{$w$}
        \State $s_{anchor} \gets$ first $s \in S$ s.t. $r_{decision} < s.round \land s \neq$ \texttt{Skip}
        \If{$s_{anchor} = \texttt{Commit}(b_{anchor})$}
        \State $b_{leader} \gets$ \Call{GetLeaderBlock}{$w$, \texttt{leaderOffset}}
        \If{$\Call{Link}{b_{anchor}, b_{leader}}$} \Return $\texttt{Commit}(b_{leader})$
            \Else \; \Return \texttt{Skip}
            \EndIf
            \EndIf
            \State \Return $\bot$
        \EndProcedure

    \end{algorithmic}
\end{algorithm}

%% file: algorithms/helper.tex
\begin{algorithm}[t]
    \algsize
    \caption{Helper Functions}
    \label{alg:helper}
    \begin{algorithmic}[1]

        \State \texttt{nodes} \Comment{The set of nodes}
        \Statex

        \Procedure{GetDecisionBlocks}{$w$}
        \State $r_{decision} \gets $\Call{DecisionRound}{$w$}
        \State \Return $DAG[r_{decision}]$
        \EndProcedure
        \Statex

        \Procedure{GetLeaderBlock}{$w, rank$}
        \State $r_{propose} \gets $\Call{ProposeRound}{$w$}
        \State $leader \gets \texttt{nodes}[(r_{propose} + rank) \bmod |\texttt{nodes}|]$
        \If{$\exists b \in DAG[r_{propose}] \text{ s.t. } b.author = leader$} \Return $b$
        \Else \; \Return $\bot$
        \EndIf
        \EndProcedure
        \Statex

        \Procedure{IsLink}{$b_{new}, b_{old}$}
        \State \Return $\exists$ a sequence of $k \in \mathbb{N}$ blocks $b_1, \ldots, b_k$ s.t.
        \Statex \hspace{2em} $b_1 = b_{old} \land b_k = b_{new} \land \forall j \in [2, k] : b_j \in \bigcup_{r \geq 1} DAG[r] \land b_{j-1} \in b_j.parents$
        \EndProcedure
        \Statex

        \Procedure{LinearizeSubDags}{$L$} \label{alg:helper:linearize}
        \State $O \gets [\;]$ \Comment{Hold output sequence}
        \For{$b_{leader} \in L$}
        \State $B \gets \{b \in \bigcup_{r \geq 1} \store[r, *] \text{ s.t. } \Call{IsLink}{b, b_{leader}} \wedge b \notin O \wedge b \text{ not already output} \}$
        \For{$b \in B \text{ in any deterministic order}$}
        \State $O \gets O \; || \; b$
        \EndFor
        \EndFor
        \State \Return $O$
        \EndProcedure
    \end{algorithmic}
\end{algorithm}

%% file: sections/proofs-revision.tex
\section{Correctness}\label{sec:proofs}

In this section, we prove that \sysname guarantees the properties of atomic broadcast.

\subsection{Safety proofs}\label{sec:safety}

\begin{lemma} \label{lem:path-exists}
    If in round $R$, $f+1$ blocks from distinct replicas vote for a block $B$, then all blocks at future rounds $R'>R$ will have a path to $B$.
\end{lemma}
\begin{proof}
    We prove the lemma by induction on $R'$. The base case is $R' = R+1$. Let $B'$ be a block at round $R'$. Since $B'$ points to $f+1$ blocks at round $R$, by quorum intersection, $B'$ must point to at least one of the blocks that vote for $B$, and thus have a path to $B$.

    For the induction case, assume the lemma holds up to round $R'$ and consider the case of round $R'+ 1$. Let $B'$ be a block at round $R'+1$. By the induction hypothesis, $f+1$ blocks at round $R'$ have paths to $B$. Since $B'$ points to $f+1$ blocks from round $R'$, by quorum intersection, $B'$ must point to at least one block that has a path to $B$.
\end{proof}

\begin{lemma}\label{lem:commit-no-skip}
    If a replica directly commits some block in a slot $S$, then no other replica skips slot $S$.
\end{lemma}
\begin{proof}
    Assume by contradiction that a replica $N$ directly commits block $B$ in slot $S$ while another replica $N'$ decides to skip $S$. Let $R$ be the round of $S$. Since $N$ directly commits $B$, there exist $f+1$ votes for $B$ at $S$. Therefore, by \Cref{lem:path-exists}, all blocks at rounds $R' > R$, including the anchor of $S$, have a path to $B$ at $S$. Thus, $N'$ cannot decide to skip $S$ using the indirect decision rule. We have reached a contradiction.
\end{proof}

\begin{observation}\label{obs:agree-commit}
    If a slot $S$ is committed at two replicas, then $S$ contains the same block at both replicas.
\end{observation}
\begin{proof}
    This follows, by construction, from two facts: (1) the sequence of skeleton slots is predetermined and known to all replicas, (2) replicas propose at most one block per slot.
\end{proof}

We say that a slot $S$ is \textit{decided} at a replica $N$ if $S$ is committed or skipped. Otherwise, $S$ is \textit{undecided}.
\begin{lemma}\label{lem:consistent}
    If a slot $S$ is decided at two replicas $N$ and $N'$, then either both replicas commit $S$, or both replicas skip $S$.
\end{lemma}
\begin{proof}
    Assume by contradiction that there exists a slot $S$ such that $N$ and $N'$ decide differently at $S$. We consider a finite execution prefix and assume \textit{wlog} that $S$ is the highest slot at which $N$ and $N'$ decide differently ($\star$). Further assume \textit{wlog} that $N$ commits $S$ and $N'$ skips $S$.
    By \Cref{lem:commit-no-skip}, neither $N$ nor $N'$ could have used the direct decision rule for $S$; they must both have used the indirect rule.
    Consider now the anchor of $S$: $N$ and $N'$ must agree on which slot is the anchor of $S$, since by our assumption ($\star$) above, they make the same decisions for all slots higher than $S$, including the anchor of $S$. Let $S'$ be the anchor of $S$; $S'$ must be committed at both $N$ and $N'$. Thus, by \Cref{obs:agree-commit}, $N$ and $N'$ commit the same block $B'$ at $S'$. But then $N$ and $N'$ cannot reach different decisions about slot $S$ using the indirect decision rule, a contradiction.
\end{proof}

We have proven the consistency of replicas' commit sequences: replicas commit (or skip) the same skeleton blocks, in the same order. However, we are not done: we also need to prove that non-skeleton blocks are committed in the same order by replicas. We show this next.

\para{Causal history and commit conditions}
Consider a replica $N$. We call the \textit{causal history} of a block $B$ in $N$'s DAG, the transitive closure of all blocks referenced by $B$ in $N$'s DAG, including $B$ itself. In \sysname, a block $B$ is committed by a replica $N$ if (1) there exists a committed skeleton block $L$ in $N$'s DAG such that $B$ is in $L$'s causal history (2) all skeleton slots up to $L$ are decided in $N$'s DAG and (3) $B$ has not been committed as part of a lower slot's causal history. In this case, we say $B$ is \textit{committed at} slot $S$, or \textit{committed with} block $L$.

\begin{lemma}\label{lem:committed-same-slot}
    If a block $B$ is committed by two replicas $N$ and $N'$, then $B$ is committed at the same slot $S$, and $B$ is committed with the same skeleton block $L$, at both $N$ and $N'$.
\end{lemma}
\begin{proof}
    Let $S$ be the slot at which $B$ is committed at replica $N$, and $L$ the corresponding skeleton block in $S$, also at replica $N$. Consider now the slot $S'$ at which $B$ is committed at replica $N'$, and $L'$ the corresponding skeleton block. Assume by contradiction that $S' \ne S$. If $S' < S$, then $N$ would have also committed $B$ at slot $S'$, since by \Cref{obs:agree-commit}, they must commit the same skeleton blocks in the same slots, so $N$ could not have committed $B$ again at slot $S$; a contradiction. Similarly, if $S < S'$, then $N'$ would have already committed $B$ at slot $S$, since by \Cref{obs:agree-commit} $N$ and $N'$ must have committed the same block in slot $S$; contradiction. Thus, it must be that $S = S'$, and by \Cref{obs:agree-commit}, $L = L'$.
\end{proof}

We can now prove the main safety properties of \sysname: total order, no duplication, and no creation.
\begin{thm}[Total Order]\label{thm:total-order}
    \sysname satisfies the total order property of Atomic Broadcast.
\end{thm}
\begin{proof}
    This property follows immediately from \Cref{lem:committed-same-slot} and the fact that replicas order the causal histories of committed blocks using the same deterministic function, and commit blocks in this order.
\end{proof}

\begin{thm}[No duplication]\label{thm:no-dup}
    \sysname satisfies the no duplication property of Atomic Broadcast.
\end{thm}
\begin{proof}
    This is by construction: a block $B$ is committed as part of the causal history of a committed skeleton block only if $B$ has not been committed along with an earlier leader block (see "Causal history \& commit conditions" above). So a replica cannot commit the same block twice.
\end{proof}

\begin{thm}[No creation]\label{thm:no-create}
    \sysname satisfies the no creation property of Atomic Broadcast.
\end{thm}
\begin{proof}
    This follows from two facts: (1) replicas only include in their local DAGs, blocks that have been previously proposed by some replica, and (2) replicas only commit blocks that they have previously included in their DAGs.
\end{proof}

\subsection{Liveness proofs}\label{sec:liveness-proofs}
We now turn to liveness. We prove liveness in the partially synchronous model here, and defer the proof of liveness under the random asynchronous model to the Supplementary Material ~\cite{supp}.


\begin{lemma}\label{lem:eventually-all-references}
    If a block $B$ produced by a correct replica $N$ references some block $B'$, then both $B$ and $B'$ will eventually be included in the local DAG of every correct replica.
\end{lemma}
\begin{proof}
    Since network links are reliable, all correct replicas will eventually receive $B$ from $N$. If a correct replica $N'$ receives $B$ from $N$, but does not have $B'$ yet, $N'$ will request $B'$ from $N$; since $N$ is correct and the network links are reliable, $N$ will eventually receive $N'$'s request, send $B'$ to $N'$, and $N'$ will eventually receive $B'$. The same is recursively true for any blocks in the causal history of $B'$, so $N'$ will eventually receive all blocks in the causal history of $B'$ and thus include $B'$ in its local DAG. Since all correct replicas eventually receive $B$'s causal history, they will include $B$ in their local DAG. 
\end{proof}

We show that \sysname is live in the partially synchronous model. Intuitively, we show that, during periods of synchrony, skeleton slots are eventually committed by all correct replicas. We assume that the timeout is set to $4\cdot \Delta$.

\begin{lemma}[Round-Synchronization] \label{lem:RoundSynchronization}
    After GST, all correct replicas will enter the same round within $3 \cdot \Delta$.
\end{lemma}
\begin{proof}
    Suppose that $R$ is the highest round a correct replica $N$ is in. Then, accounting for $2 \cdot \Delta$ to request any missing parent blocks, every correct replica will add $N$'s round $R$ block to its DAG within $3\cdot\Delta$ and also enter $R$.
\end{proof}

\begin{lemma}[Leader-Proposal]\label{lem:LeaderProposal}
    After GST, the proposal of a correct skeleton replica will get $f+1$ votes.
\end{lemma}
\begin{proof}
    Let N be a correct skeleton replica proposing block $B$ in round $R$. We want to show that every replica proposing in $R+1$ votes for $B$. Let $N'$ be a replica to propose a block in $R+1$. If it references $f+1$ blocks and all skeleton blocks in round $R$, then $N'$ votes for $B$. Otherwise, it must time out in round $R$. Since $N$ enters within at most $3\Delta$ after $N'$ (\Cref{lem:RoundSynchronization}), $N'$ will receive $B$ before the timeout and again, vote for it. This shows that every replica which proposes in $R+1$ votes for $B$. 
The only case left is that of a replica $N''$ never proposing in round $R+1$. The first such $N''$ not proposing in $R+1$ can only happen if there already exist $f+1$ blocks in round $R+1$ causing $N''$ to advance past $R+1$. In this case, the lemma already holds. 
\end{proof}

\begin{lemma}[Sufficient Votes]\label{lem:SufficientVotes}
    After GST, the proposal of a correct skeleton replica will get directly committed. 
\end{lemma}
\begin{proof}
    For a replica to propose a block at the decision round, it needs to (a) get the proposal for the skeleton slot and (b) have $f + 1$ parents. By Lemma~\ref{lem:LeaderProposal}, a correct skeleton replica gets $f+1$ votes at least one of which must be from a correct replica. All other correct replicas will receive this block proposal and adopt the parents within $3 \cdot \Delta$. Since timeouts are set to $4 \cdot \Delta$, replicas will receive $f+1$ decision blocks and can commit the skeleton replica's proposal.
\end{proof}

\begin{revblock}
\begin{lemma}\label{lem:Liveness}
    After GST, any undecided slot eventually gets decided.
\end{lemma}
\begin{proof}
    Let $S$ be an undecided slot in the highest round $R$, so that every slot above $S$ is decided. After GST, due to the scheduler (\eg round-robin), there will eventually be a correct skeleton replica for a slot $S' > S$ in a round $k > R+1$.
    By Lemma~\ref{lem:SufficientVotes}, the correct skeleton replica's proposal at $S'$ will receive $f + 1$ votes and become committed. Therefore,  the anchor of $S$ is a committed slot, and the indirect decision rule classifies $S$ as \texttt{commit} or \texttt{skip}. The same argument applied in decreasing round order decides the remaining slots, so every undecided slot is eventually decided.
\end{proof}

\begin{lemma}\label{lem:gst-all-commit}
    After GST, every correct replica's skeleton slot will be committed and included in the commit sequence.
\end{lemma}
\begin{proof}
    By \Cref{lem:LeaderProposal}, the proposal of a correct skeleton replica will get $f + 1$ votes, and by our commit rule, this slot will be directly or indirectly committed. By \Cref{lem:Liveness}, all previous undecided slots are eventually decided, so the commit sequence can be extended through this slot.
\end{proof}

\begin{thm}[Validity]
    \sysname satisfies the validity property of Atomic Broadcast.
\end{thm}
\begin{proof}
    Let $N$ be a correct replica and $B$ a block proposed by $N$. We show that $B$ is eventually committed by every correct replica.
    By \Cref{lem:eventually-all-references}, $B$ is eventually included in the local DAG of every correct replica.
    So every correct replica will eventually include a reference to $B$ in at least one of its blocks. Let $R$ be the highest round at which some correct replica includes a reference to $B$ in one of its blocks. By \Cref{lem:gst-all-commit}, eventually some skeleton block $B'$ at a round $R' > R+1$ will be committed. Block $B'$ must reference at least $f+1$ blocks. Since all replicas have $B$ in their causal histories by round $R$, $B'$ must therefore have a path to $B$. \Cref{lem:Liveness} guarantees that all slots before $B'$ are eventually decided, so $B'$ is eventually committed. Thus, $B$ will be committed at all correct replicas at the latest when $B'$ is committed along with its causal history.
\end{proof}

\begin{thm}[Agreement]
    \sysname satisfies the agreement property of Atomic Broadcast.
\end{thm}
\begin{proof}
    Let $N$ be a correct replica and $B$ be a block committed by $N$. We show that $B$ is eventually committed by every correct replica. Let $l$ be the skeleton block with which $B$ is committed, and $S$ the corresponding slot. By \Cref{lem:Liveness}, all slots up to and including $S$ are eventually decided by all correct replicas. By \Cref{obs:agree-commit}, all correct replicas commit $l$ in $S$. Therefore, all correct replicas commit $B$ eventually.
\end{proof}

\end{revblock}

%% file: sections/implementation.tex
\section{Implementation} \label{sec:implementation}

We implemented \sysname in Rust, building on the open-source \mysticeti codebase~\cite{mysticeti-code}, which contains approximately $14{,}000$ lines of code (LOC). Our implementation uses \texttt{tokio}~\cite{tokio} for asynchronous networking and raw TCP sockets for replica-to-replica communication.

To support crash recovery and ensure data persistence, we implemented a Write-Ahead Log (WAL) tailored to \sysname. This provides stronger resilience than existing implementations of Multi-Paxos~\cite{paxos-raft-code} and QuePaxa~\cite{quepaxa-code}, which lack mechanisms for recovering after crashes, though this absence arises from implementation choices rather than protocol specifications.
Following prior implementations of Multi-Paxos, QuePaxa, and Rabia~\cite{rabia-code}, we support batching at the replica level to amortize communication costs by packing multiple commands into a single message.
We collocated front- and back-end within the same \sysname binary.
Each front-end communicates with its assigned back-end, a design choice common in existing consensus protocol implementations~\cite{epaxos-code-modified}.

We also implement a benchmark framework to simulate performance under a random asynchronous network model. 
The framework forces each consensus instance to gather a randomly selected quorum to commit rather than rely on whichever majority returns first. 
This shifts performance away from the quickest or geographically closest replicas and instead drives execution through a random message-delivery schedule that removes any bias toward low-latency quorums. 
We use this framework only in \Cref{subsec:eval-random} and rely on standard communication patterns for the rest of our evaluation in \Cref{sec:evaluation}.
To the best of our knowledge, this is the first work to introduce an evaluation framework for measuring protocol performance under the random asynchronous model, and we present this evaluation setup as a novel contribution of independent interest.
We will open-source our \sysname implementation upon acceptance to support reproducibility and enable further research in this area.

%% file: sections/evaluation.tex
\section{Evaluation}\label{sec:evaluation}

Our evaluation compares performance, scalability, and resource efficiency to existing CFT consensus protocols. 
We experiment across a range of execution conditions, including favorable (fault-free, timely communication), crash-fault scenarios, and both normal and unstable network conditions.

\input{sections/evaluation/baselines}
\input{sections/evaluation/setup}
\input{sections/evaluation/normal-case}
\input{sections/evaluation/scalability-num-replicas}
\input{sections/evaluation/scalability-payload}
\input{sections/evaluation/crash}
\input{sections/evaluation/random-async}
\input{sections/evaluation/bft}
\input{sections/evaluation/cpu}

%% file: sections/evaluation/baselines.tex
\subsection{Baselines}\label{subsec:experiment_related}

We compare the performance of \sysname against state-of-the-art CFT protocols:
Multi-Paxos~\cite{lamport2001paxos, paxos-raft-code}, EPaxos~\cite{moraru2013there, epaxos-code}, QuePaxa~\cite{tennage2023quepaxa, quepaxa-code},
Rabia~\cite{pan2021rabia, rabia-code},
and SADL-RACS~\cite{tennageracs}

We select these baselines for their contrasting designs, ensuring coverage of a broad spectrum of protocol architectures. 
Multi-Paxos uses a classic leader-based approach that channels all requests through a single leader. 
EPaxos adopts a multi-leader design in which every replica can propose and commit commands concurrently under dependency constraints. 
EPaxos produces only a partial order, a weaker abstraction than the total order that \sysname enforces. 
QuePaxa represents the current state of the art in randomized consensus: it matches Multi-Paxos under synchronous conditions and switches to a fallback mode that preserves throughput under asynchrony, which makes it well suited for wide-area deployments. 
Rabia uses randomized techniques to optimize for low-latency local-area settings; we include it for completeness. 
SADL-RACS achieves high throughput by separating command dissemination from the consensus critical path.
We include QuePaxa, Multi-Paxos, EPaxos, and \sysname in all subsequent experiments. 
Due to page-limit constraints, we report only best-case evaluations for SADL-RACS and Rabia.
We observed that the existing EPaxos codebase~\cite{epaxos-code} does not support more than five replicas, an implementation limitation noted by prior work~\cite{tennage2023quepaxa, tennageracs}.
Therefore, we omit using EPaxos in the scalability experiment in ~\Cref{subsec:eval-scala-numreplicas}.
\sysname is designed for crash fault-tolerant systems; however, for completeness, we also compare \sysname against two state-of-the-art DAG-based, partially synchronous BFT consensus protocols: Bullshark~\cite{bullshark} and Mysticeti~\cite{mysticeti2023}.
We select these two BFT protocols because they share key architectural features with \sysname and because both are widely deployed in production systems~\cite{sui-lutris,sui-code}.

%% file: sections/evaluation/setup.tex
\subsection{Experimental setup}\label{subsec:experiment_setup}

We deploy replicas on Amazon EC2 virtual machines~\cite{amazon-ec2} of type \texttt{c5.4xlarge}.
Each machine provides 10\,Gbps of bandwidth, 16 virtual CPUs on a 2.5\,GHz Intel Xeon Platinum 8124M, 32\,GB memory, and runs Ubuntu Linux 22.04~LTS~\cite{ubuntu}.
We select these machines because they provide decent performance, are in the price range of ``commodity servers'', and are in line with recent related work~\cite{mysticeti2023,mahi-mahi}.
We conducted experiments in a wide-area network (WAN) with replicas distributed across multiple AWS regions: Cape Town (\texttt{af-south-1}), Hyderabad (\texttt{ap-south-2}), Jakarta (\texttt{ap-southeast-3}), Osaka (\texttt{ap-northeast-3}), and Milan (\texttt{eu-south-1}).

We instantiate several geo-distributed benchmark clients collocated with each replica, submitting transactions in an open loop model~\cite{schroeder2006open}, at a fixed rate.
We experimentally increase the submission rate and record the throughput and latency of commits.
As a result, all plots illustrate the steady-state latency of all systems under load, as well as the maximal throughput they can provide, after which latency increases rapidly.
We verified that the CPU, memory, bandwidth, and storage resources are sufficient to host both the replica and client for all protocols under test.

Following prior studies~\cite{pan2021rabia, tennage2023quepaxa} and production systems~\cite{bronson2013tao}, we set the command size to 18~bytes.
We set the timeouts at 5\,sec in each protocol.
For each protocol, we measure end-to-end commit latency.
Throughput is measured in commands per second (cmd/sec), where one command corresponds to a single 18-byte request. 
All protocols use batching to amortize network latency.

%% file: sections/evaluation/normal-case.tex
\subsection{Normal case WAN performance}\label{subsec:eval-normal}

\begin{figure}[t]
    \centering
    \includegraphics[width=\columnwidth]{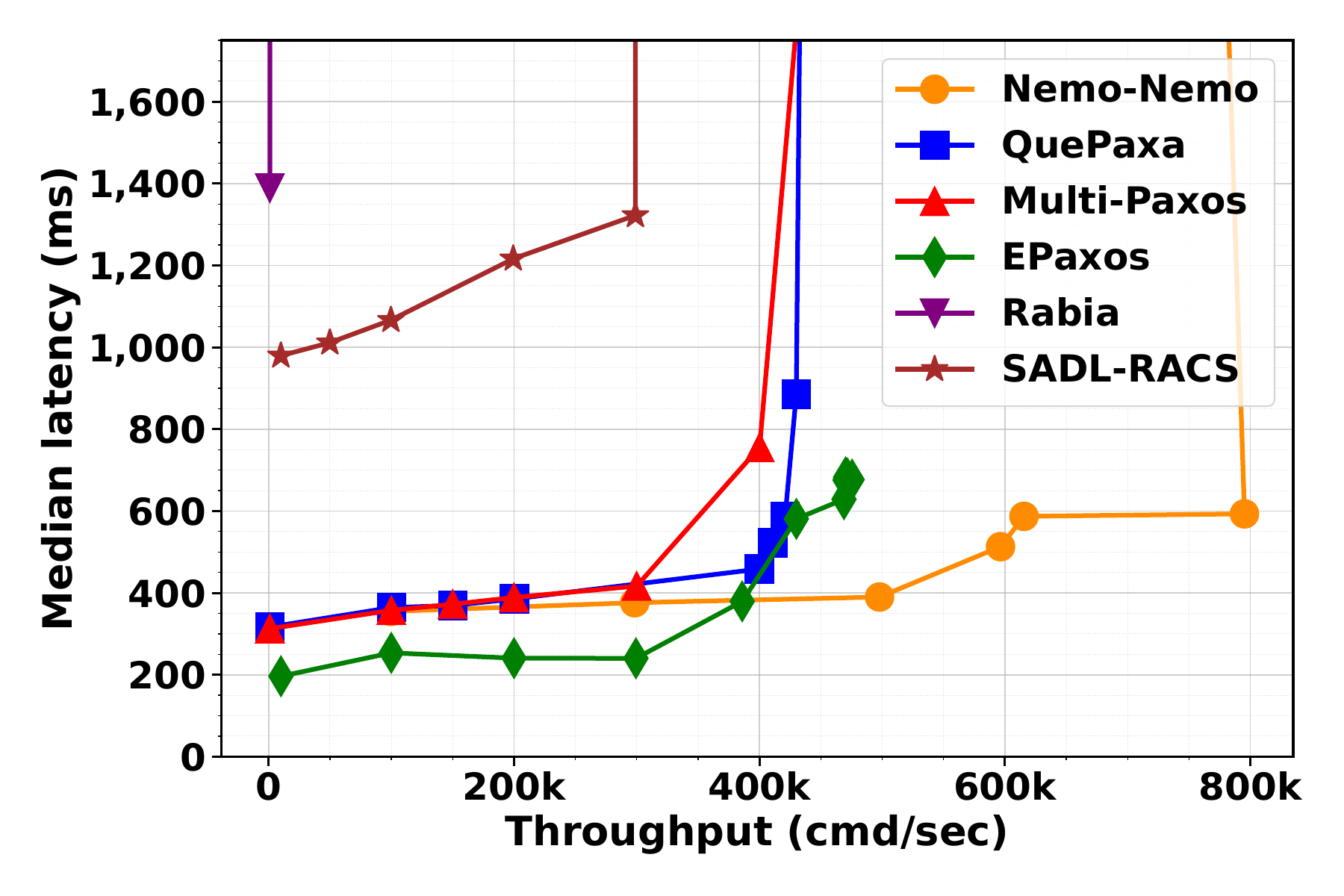}
    \vskip -1em
    \caption{
        Performance under normal case WAN execution.
    }
    \label{fig:best_case_5_replica_throughput_latency}
\end{figure}

We evaluate the normal-case performance of \sysname in a WAN under favorable (``synchronous'') network conditions. \Cref{fig:best_case_5_replica_throughput_latency} shows the throughput and median latency.

\para{Multi-Paxos}
We observe in \Cref{fig:best_case_5_replica_throughput_latency} that \sysname reaches a saturation throughput of 800k\,cmd/s at a median latency of 600\,ms.
In contrast, Multi-Paxos reaches only 400k\,cmd/s at 750\,ms.
This corresponds to a 2x throughput improvement and a 20\% latency reduction for \sysname over Multi-Paxos.

\sysname's throughput advantage over Multi-Paxos stems from two key factors.
First, Multi-Paxos funnels all client traffic through a single leader, creating a bottleneck, whereas \sysname distributes load uniformly across all replicas (see \Cref{sec:dissemination}).
Second, \sysname amortizes a single message per batch of commands (see \Cref{sec:consensus}), while Multi-Paxos requires ten messages per batch in a five-replica deployment, consisting of five \textit{propose} messages followed by five \textit{accept} messages.

\para{EPaxos}
We observe in \Cref{fig:best_case_5_replica_throughput_latency} that EPaxos delivers over 400k\,cmd/s at a median latency of 440\,ms, compared to \sysname which reaches 800k\,cmd/s at 600\,ms.
Both EPaxos and \sysname distribute load among all participating nodes and therefore outperform leader-based Multi-Paxos in latency.

For a throughput below 300k\,cmd/s, EPaxos achieves a 130\,ms lower median latency than \sysname.
This latency advantage arises because EPaxos provides only a partial order of commands, which incurs less coordination overhead than \sysname, which enforces a total order.
In this experiment, we configured EPaxos with a 2\% conflict rate~\cite{tollman2021epaxos}, meaning that 98\% of commands commit in a single round trip, resulting in lower commit latency than \sysname at moderate load.

However, beyond 400k\,cmd/s, \sysname achieves roughly 2× the throughput of EPaxos.
This gap stems from two factors.
(1) The EPaxos implementation suffers from implementation limitations—most notably a single-threaded main execution path—whereas \sysname employs a highly parallelized design.
(2) Under high arrival rates, EPaxos’s conflict-resolution overhead becomes a dominant bottleneck.
By avoiding specialized conflict resolution, \sysname sustains substantially higher throughput.

EPaxos is well suited for applications requiring only partial order with low conflict rates, such as key-value stores where per-key ordering is sufficient.
By contrast, \sysname targets applications requiring a total order of all commands, providing higher throughput for workloads that demand strong consistency.

\para{QuePaxa}
We observe in \Cref{fig:best_case_5_replica_throughput_latency} that QuePaxa achieves only 400k\,cmd/s at a median latency of 460\,ms. Under normal-case executions, QuePaxa\footnote{
    Ideally, both QuePaxa and Multi-Paxos should incur the same overhead and achieve comparable performance in normal-case operation.
    However, the QuePaxa implementation yields higher throughput than Multi-Paxos because it leverages gRPC-based high-speed messaging and multi-threading, whereas the Multi-Paxos implementation employs single-threaded execution.}
funnels all commands through a single leader replica.
As a result, QuePaxa's performance is limited by the leader replica's capacity, whereas \sysname achieves higher throughput and lower latency by load balancing across all participating nodes.

\para{SADL-RACS}
SADL-RACS uses a mempool to disseminate client commands off the consensus critical path.
In our deployment, it achieves a saturation throughput of 300k\,cmd/s with a median latency of 1.35\,s.
In comparison, \sysname reaches 800k\,cmd/s and 600\,ms, corresponding to a 166\% higher throughput and 125\% lower latency.

The performance advantage of \sysname arises from its avoidance of extra dissemination round trips.
SADL-RACS requires two additional hops to propagate commands, which significantly increases latency.
By embedding command dissemination into the DAG and balancing load across replicas, \sysname sustains higher throughput and lower latency.

\para{Rabia}
In our WAN deployment, Rabia delivers only 1k\,cmd/s with latency above 1.25\,s, whereas \sysname achieves 800k\,cmd/s at 600\,ms.
Its lower WAN performance stems from its design assumptions and optimizations for low-latency, high-bandwidth LANs.
As network diameter and latency variability increase, Rabia's design assumptions no longer hold, leading to severely reduced throughput and higher latency.

\para{Tail latency}
None of the evaluated systems is optimized for tail latency, and in practice we observe comparable percentile 99 behavior across all protocols.
Under a 1\,s percentile 99 bound, \sysname sustains 400k\,cmd/s, while EPaxos reaches 375k\,cmd/s, and both QuePaxa and Multi-Paxos reach roughly 300k\,cmd/s.
These results are consistent with prior findings reported in related work\cite{tennage2023quepaxa}.

\para{Summary}
Across our experiments, \sysname consistently outperforms existing CFT protocols in WAN deployments, achieving higher throughput and lower latency.
Its advantages arise from uniform load distribution, minimal message delays, and avoidance of specialized assumptions (such as latency bounds) or conflict-resolution bottlenecks.
These observations validate our vision: \sysname is well suited for applications that require total ordering, high throughput, and low latency, in cloud and wide-area environments.

%% file: sections/evaluation/scalability-num-replicas.tex
\subsection{Scalability with number of replicas}\label{subsec:eval-scala-numreplicas}

\begin{figure}[t]
    \centering
    \includegraphics[width=\columnwidth]{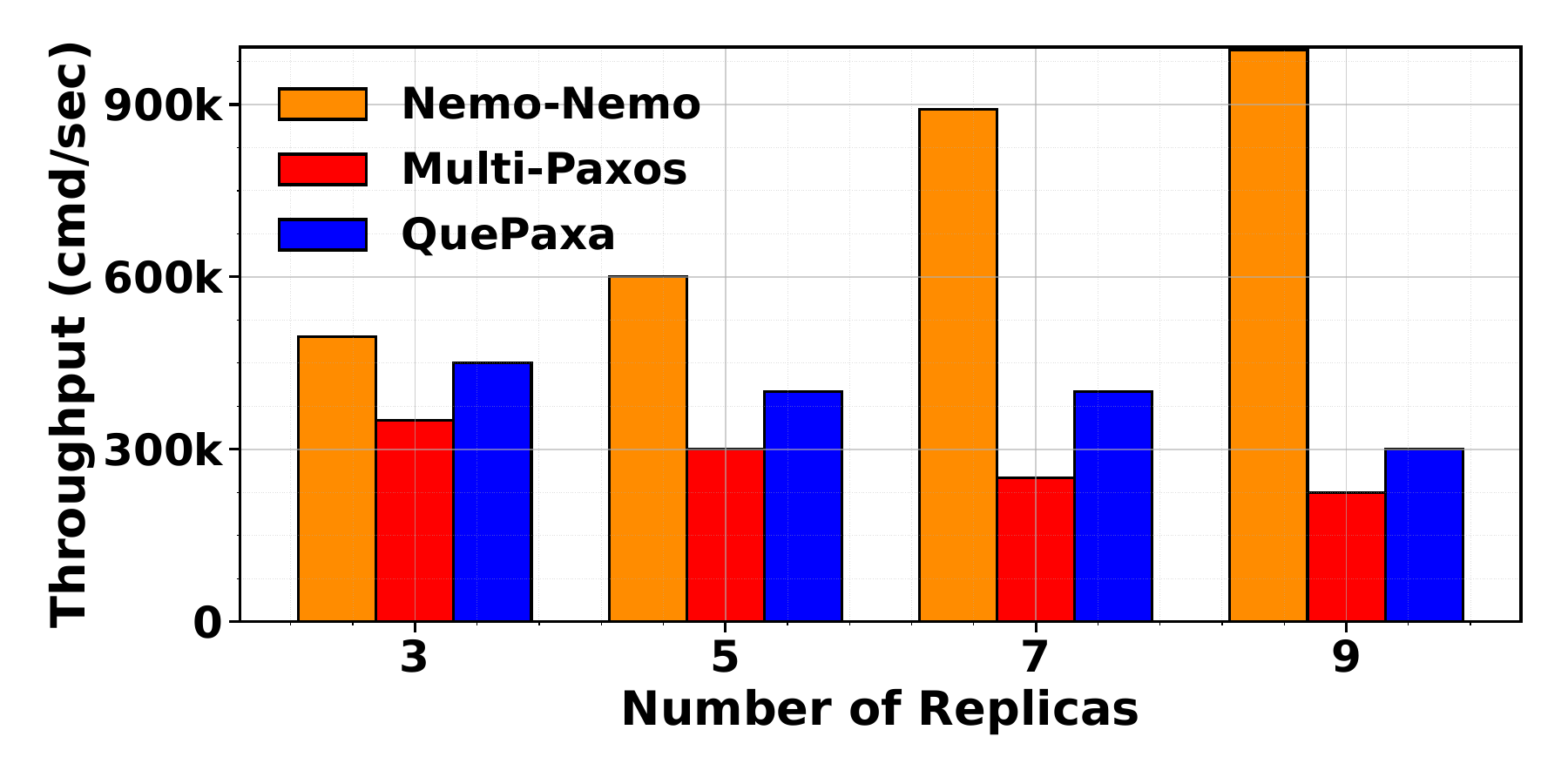}
    \vskip -1em
    \caption{
        Scalability with replication factor. Saturation throughput under 500\,ms median latency.
    }
    \label{fig:scalability_num_replicas}
\end{figure}

We evaluate how \sysname scales as the number of replicas increases.
We deploy ensembles of 3, 5, 7, and 9 replicas across geographically distributed AWS regions and compare against pipelined Multi-Paxos and QuePaxa\footnote{Due to an implementation error, EPaxos only supports up to 5 replicas and is excluded from this experiment.}.
\Cref{fig:scalability_num_replicas} shows the saturation throughput under 500\,ms median latency.

\para{Leader-based performance degradation}
Multi-Paxos and QuePaxa exhibit throughput degradation of 35\% and 33\%, respectively, when scaling from 3 to 9 replicas.
The performance degradation stems from leader-centric designs: as the replication factor grows, the leader must handle more messages due to larger quorum sizes, creating a resource bottleneck.

\para{Multi-leader scalability}
In contrast, \sysname's throughput doubles from 500k to 1M\,cmd/sec over the same range, demonstrating superior scalability.
\sysname exhibits opposite behavior due to its multi-leader DAG-based architecture.
With only 3 replicas, \sysname underutilizes available network and CPU resources.
Adding more replicas improves resource multiplexing across leaders, increasing throughput.
Furthermore, more replicas increase the number of blocks per round and the causal history size of leader blocks without additional network hops, enabling higher parallelism.

%% file: sections/evaluation/scalability-payload.tex
\subsection{Scalability with command size}\label{subsec:eval-scala-payload}

\begin{figure}[t]
    \centering
    \includegraphics[width=\columnwidth]{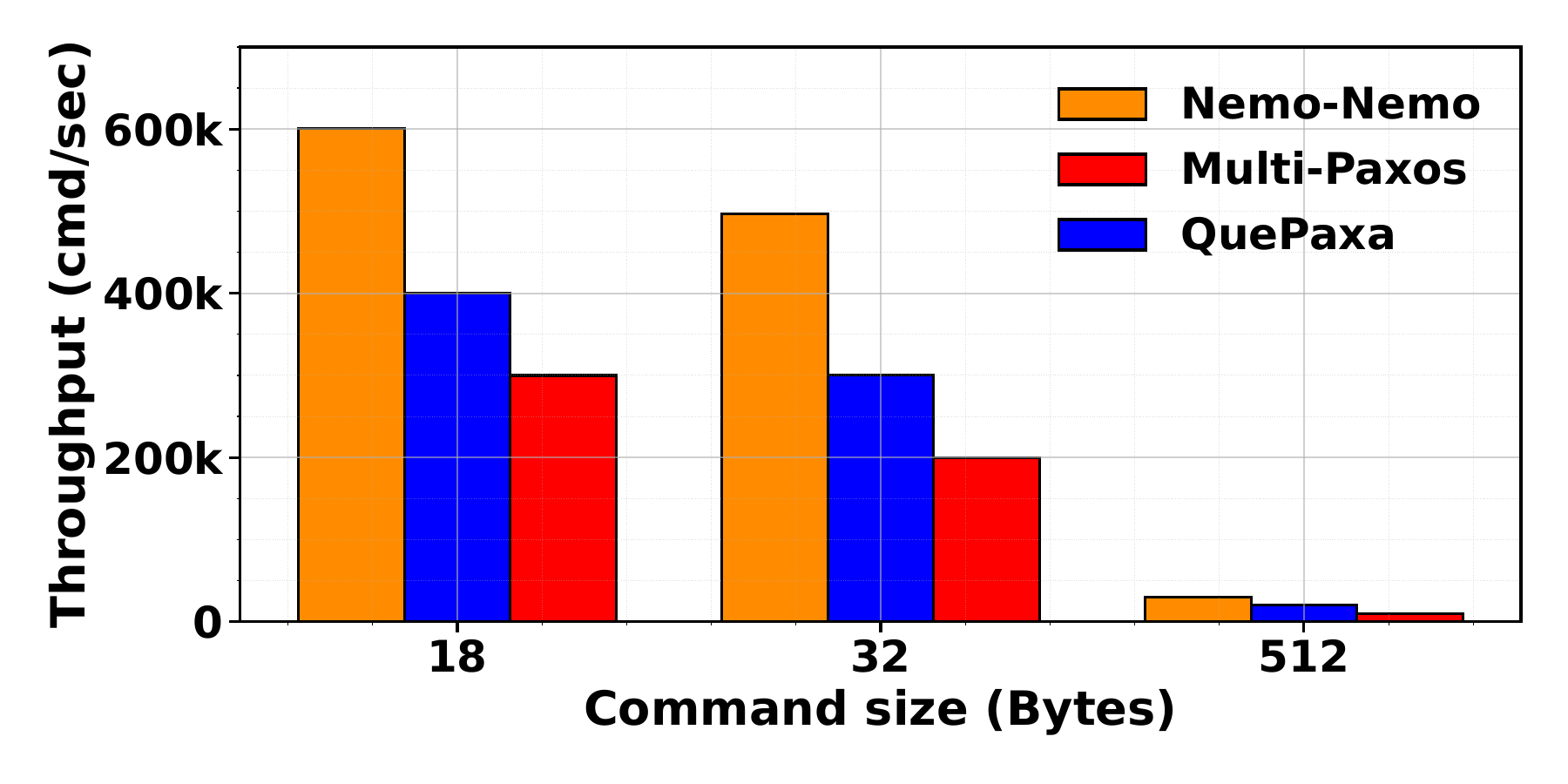}
    \vskip -1em
    \caption{
        Scalability with command size. Saturation throughput under 500\,ms median latency and 5 replicas.
    }
    \label{fig:scalability_message_size}
\end{figure}

We evaluate the impact of command size on \sysname's performance.
We experiment with three command sizes: 18\,B, 32\,B, and 512\,B, as used in recent SMR work~\cite{alimadadi2023waverunner,tennageracs}.
We deploy \sysname, pipelined Multi-Paxos, and pipelined QuePaxa in a WAN setting with 5 replicas. \Cref{fig:scalability_message_size} depicts the saturation throughput under 500\,ms median latency.\footnote{The EPaxos implementation does not allow configurable command sizes and was omitted from this experiment.}

As command size increases from 18\,B to 512\,B, Multi-Paxos degrades from 300k to 10k\,cmd/s, while QuePaxa drops from 400k to 20k\,cmd/s.
In contrast, \sysname maintains 50--66\% higher throughput at 18\,B and 32\,B, and still retains a 5\% advantage at 512\,B, degrading from 600k to 30k\,cmd/s.

This performance gap stems from fundamental architectural differences.
In Multi-Paxos and QuePaxa, the single leader must handle all client commands and broadcast them to followers, making the leader's network bandwidth the bottleneck as command size grows.
\sysname's multi-leader DAG-based design distributes bandwidth usage evenly across all replicas, avoiding the leader bottleneck entirely.

%% file: sections/evaluation/crash.tex
\subsection{Performance under crash failures}\label{subsec:eval-crash}

\begin{figure*}[t]
    \centering
    \begin{subfigure}[b]{0.24\linewidth}
        \centering
        \includegraphics[width=\textwidth]{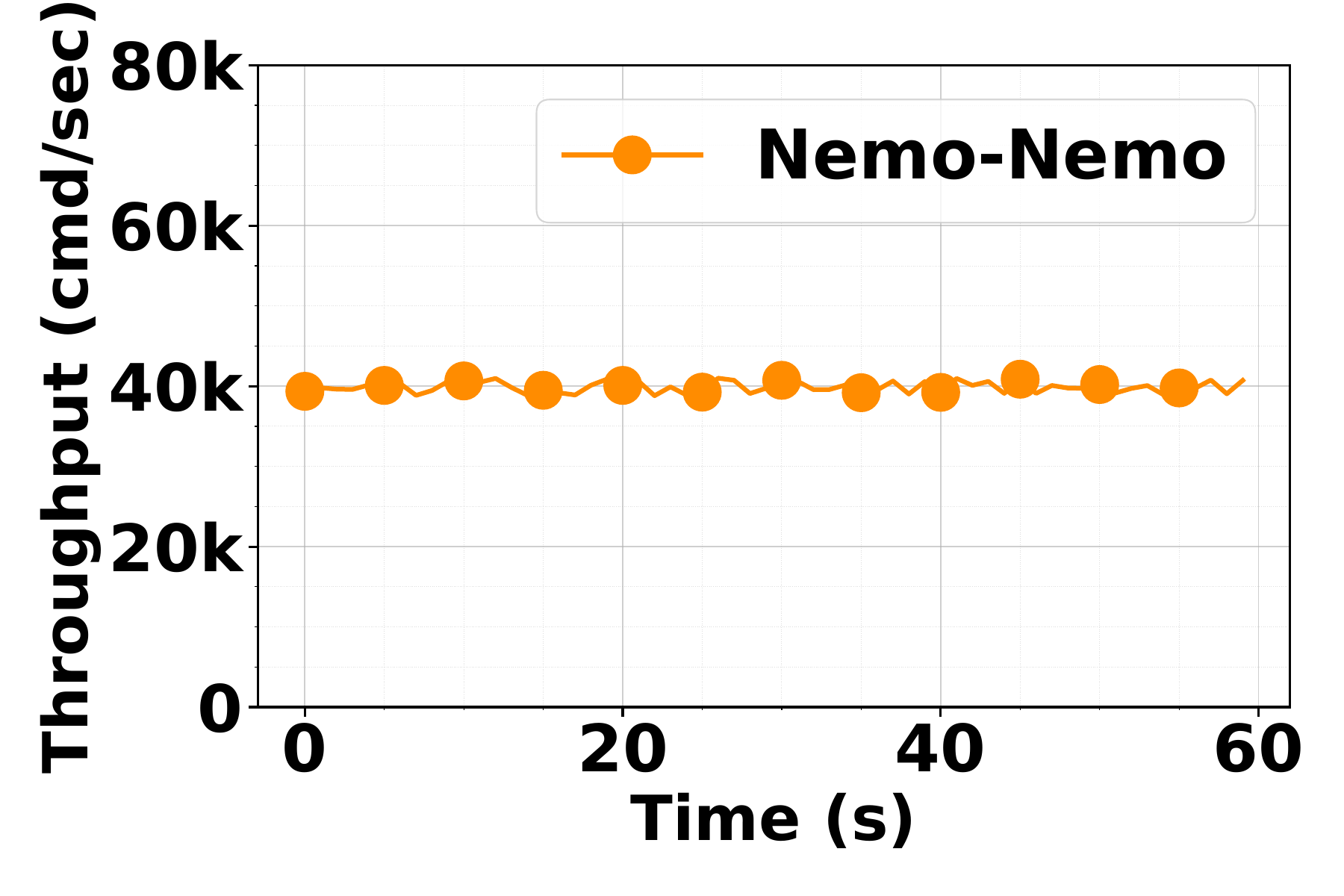}
        \vskip -1em
        \label{fig:nemo_crash}
    \end{subfigure}
    \begin{subfigure}[b]{0.24\linewidth}
        \centering
        \includegraphics[width=\textwidth]{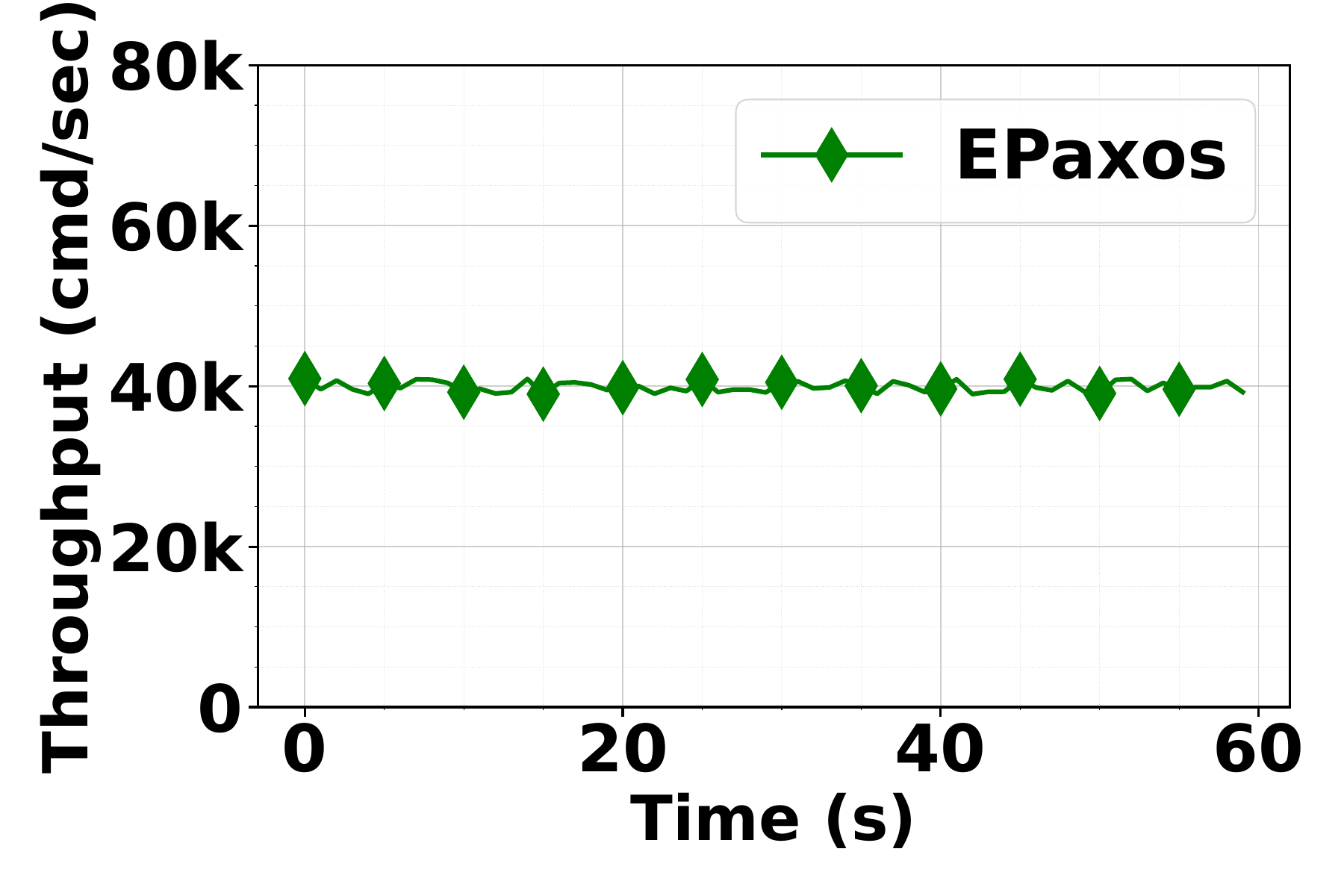}
        \vskip -1em
        \label{fig:epaxos_crash}
    \end{subfigure}
    \begin{subfigure}[b]{0.24\linewidth}
        \centering
        \includegraphics[width=\textwidth]{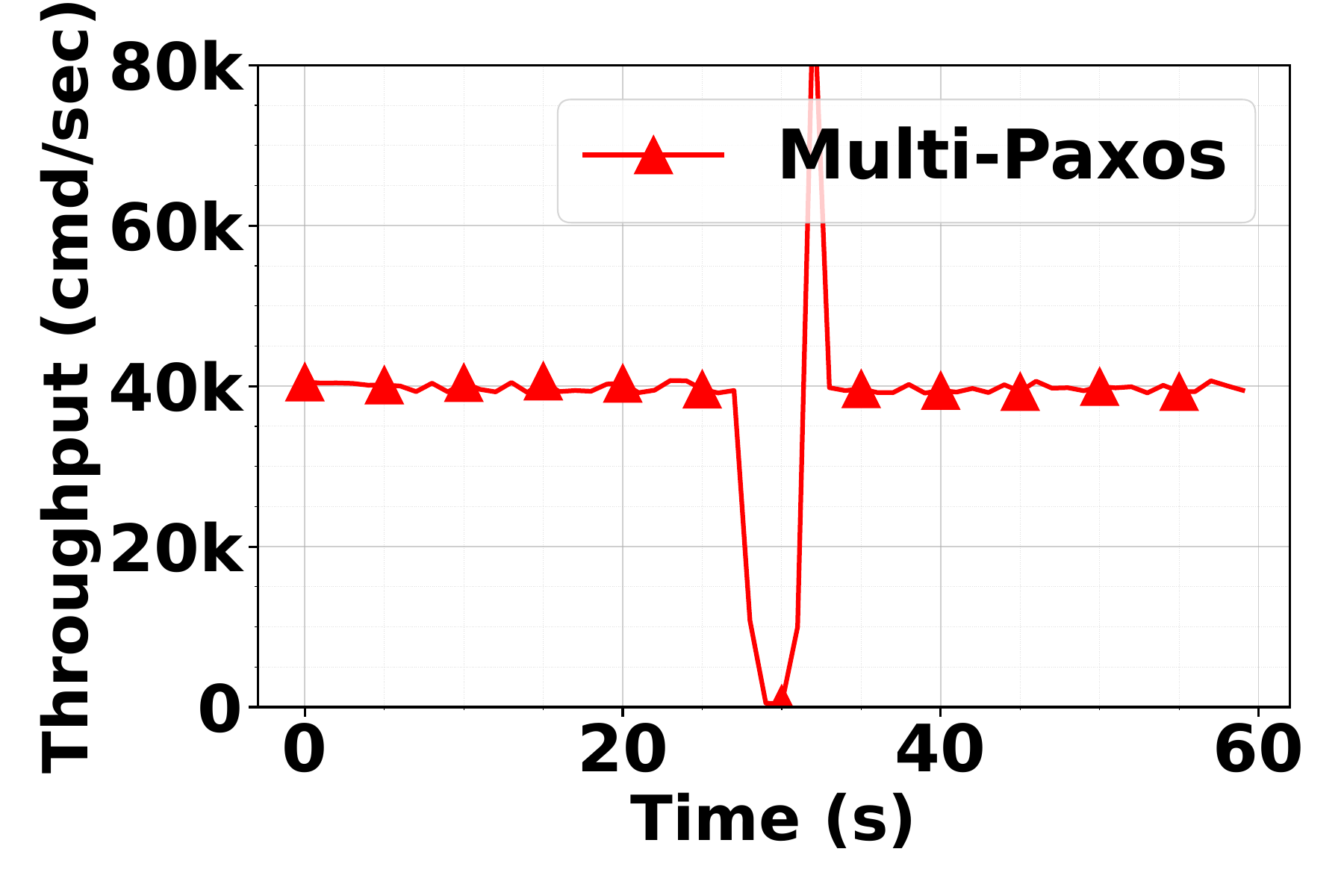}
        \vskip -1em
        \label{fig:paxos_crash}
    \end{subfigure}
    \begin{subfigure}[b]{0.24\linewidth}
        \centering
        \includegraphics[width=\textwidth]{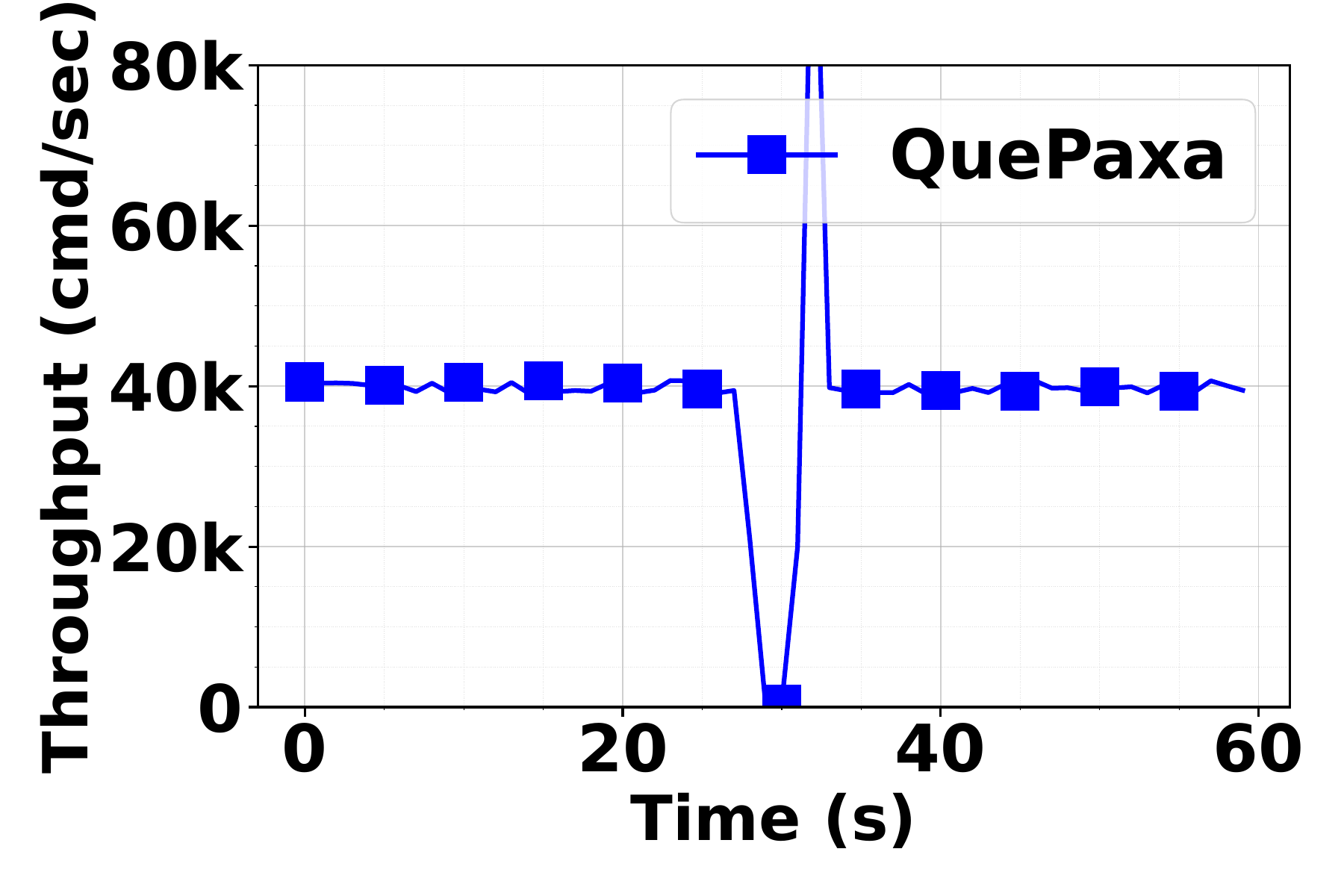}
        \vskip -1em
        \label{fig:quepaxa_crash}
    \end{subfigure}

    \caption{
        Throughput under crash failures with 5 replicas.
        At 25 seconds, we crash the leader in Multi-Paxos and QuePaxa, and a random replica in \sysname and EPaxos.
        Constant arrival rate of 40k\,cmd/sec.
    }
    \label{fig:crash_throughput}
\end{figure*}

We evaluate how each system handles crash failures by deploying five replicas in a WAN setup under a constant load of 40k\,cmd/sec.
At 25 seconds, we crash the leader replica in Multi-Paxos and QuePaxa, and a random replica in \sysname and EPaxos.
\Cref{fig:crash_throughput} shows the throughput over time.

\para{Leader-based protocols suffer downtime}
Multi-Paxos throughput drops to zero immediately upon leader crash and remains at zero during the leader election (view change).
Once a new leader is elected, the backlog of accumulated commands is processed, causing a throughput spike before returning to steady state.
QuePaxa exhibits similar behavior despite employing multiple leaders via a hedging schedule: by default, only the first leader proposes in the crash-free case, and the second leader begins proposing only \emph{after the configured hedging timeout}.
This design still creates a period of zero throughput during leader transition, followed by a spike when the new leader processes the backlog.

\para{Multi-leader protocols avoid downtime}
Both EPaxos and \sysname exhibit no visible downtime under replica crashes.
As multi-leader protocols, the failure of one replica does not affect the ability of the remaining replicas to make progress.
\sysname avoids explicit view changes and promptly discards blocks from crashed replicas, preventing head-of-line blocking in the commit sequence.
This ensures that \sysname remains robust. Even if a crashed replica was previously a leader, its uncommitted blocks do not stall the system's progress.

%% file: sections/evaluation/random-async.tex
\subsection{Performance under random asynchrony}\label{subsec:eval-random}

\begin{figure}[t]
    \centering
    \includegraphics[width=\columnwidth]{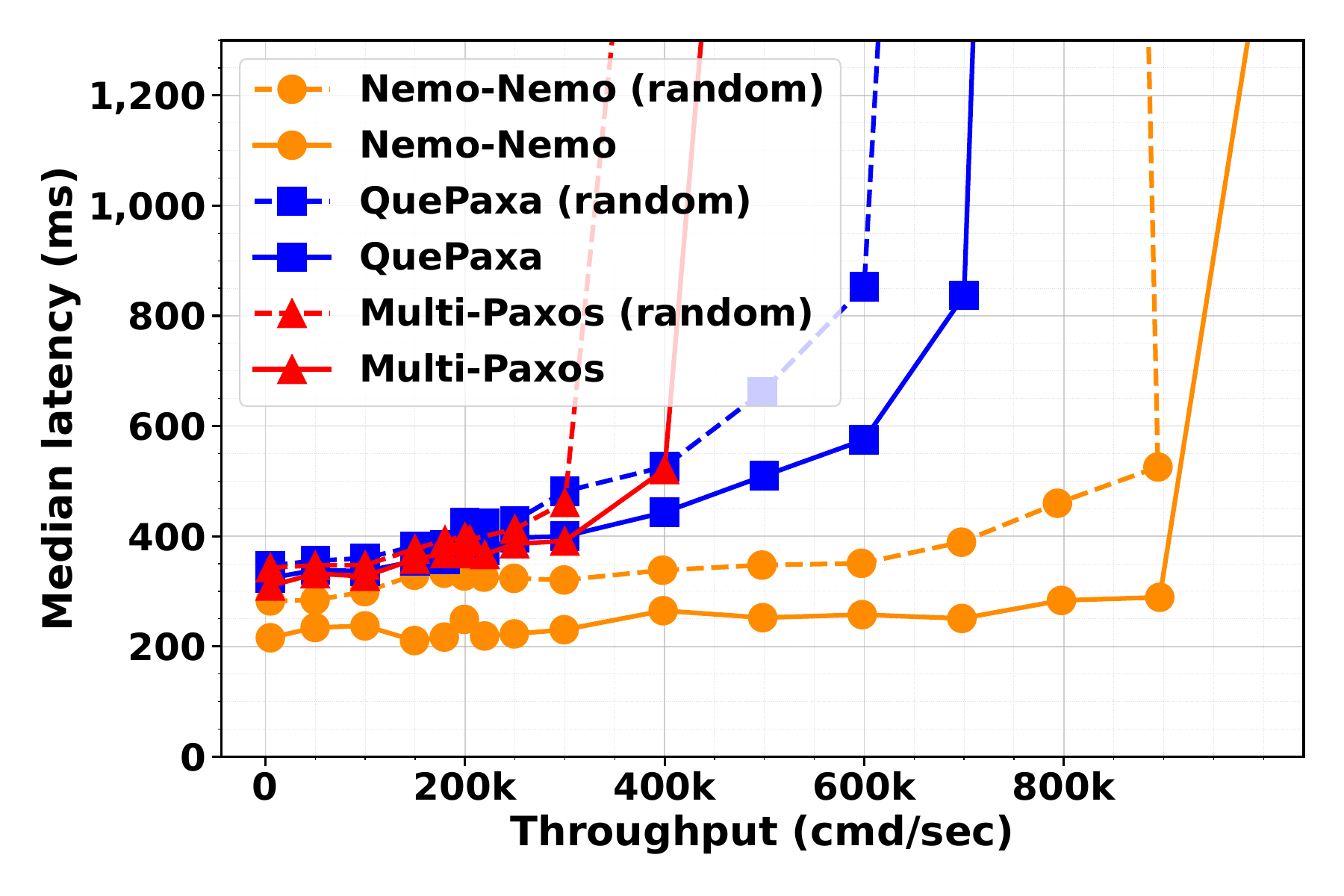}
    \vskip -1em
    \caption{
        Performance under normal conditions and random asynchrony with 5 replicas.
    }
  \label{fig:random_case_5_replica_throughput_latency_new_setup}
\end{figure}

We evaluate \sysname's performance when experiencing random message delays, leveraging the benchmarking framework presented in \Cref{sec:implementation}.
To amplify the effects of network randomization, we deploy 5 replicas in Cape Town (af-south-1), London (eu-west-2), Ireland (eu-west-1), Milan (eu-south-1), and Hyderabad (ap-south-2).
This setup includes three replicas in Europe (forming a fast majority) and two geographically distant replicas, creating significant latency variance.
\Cref{fig:random_case_5_replica_throughput_latency_new_setup} compares throughput and median latency for both normal conditions and random asynchrony.

\para{Single leader-based protocols suffer under asynchrony}
Multi-Paxos degrades from 400k to 300k\,cmd/sec while
QuePaxa drops from 500k to 400k\,cmd/sec, under 500ms latency bound.
This degradation stems from their reliance on single leader-based dissemination: all commands must flow through a single leader to its quorum.
Under random asynchronous network conditions, the leader or its quorum members experience slow message delivery, creating a bottleneck that degrades overall system performance.

\para{DAG-based dissemination provides robustness}
In stark contrast, \sysname maintains stable throughput at 900k\,cmd/sec, showing no throughput degradation despite network randomization.
This results in \sysname outperforming Multi-Paxos by 3$\times$ and QuePaxa by 2.25$\times$, in throughput, under random asynchrony.
\sysname's robustness arises from its DAG-based architecture: all replicas disseminate commands concurrently, creating multiple parallel paths to commitment.
\sysname adapts to network conditions without performance degradation.

%% file: sections/evaluation/bft.tex
\subsection{Comparison against BFT protocols}\label{subsec:eval-bft}

\begin{figure}[t]
    \centering
    \includegraphics[width=\columnwidth]{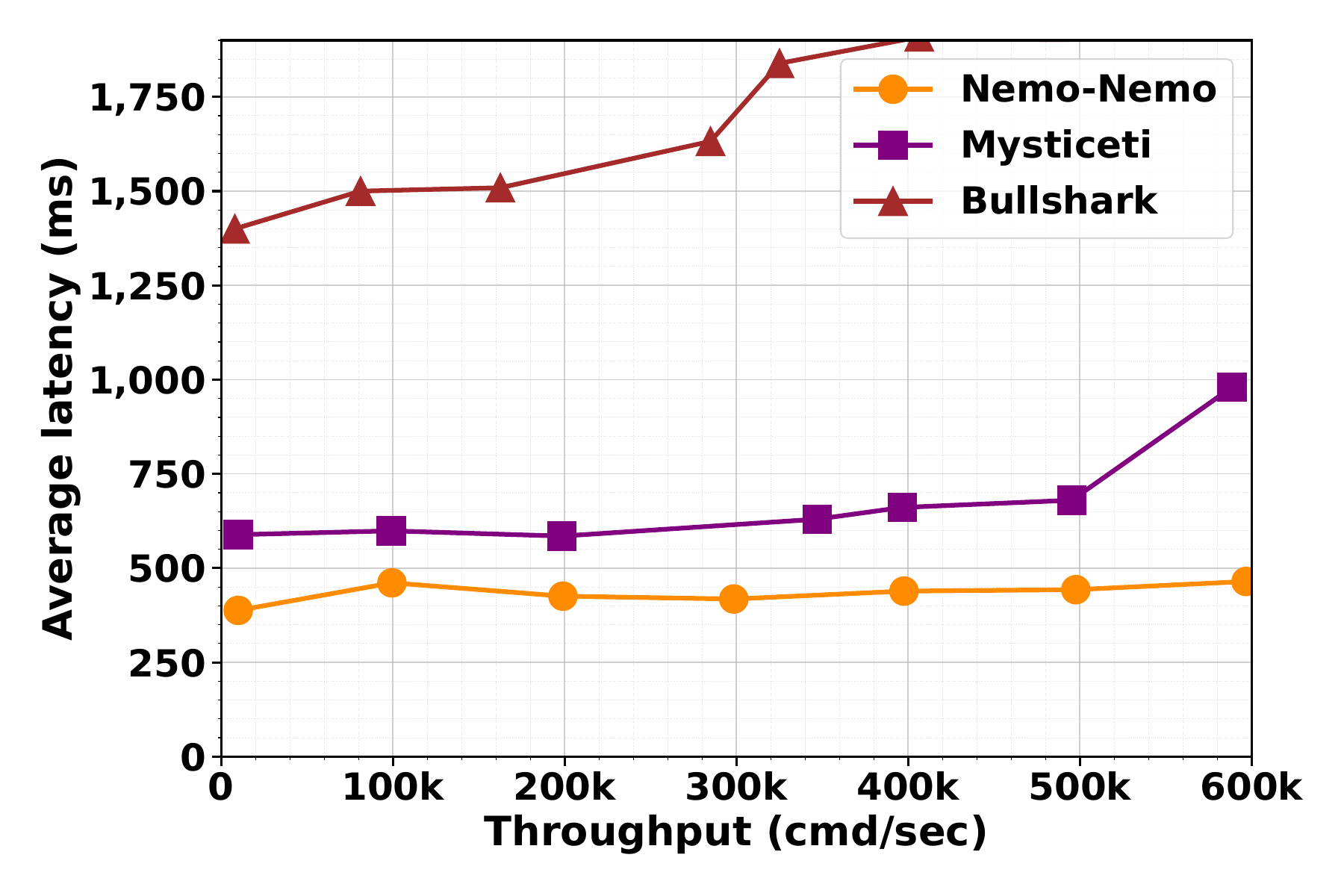}
    \vskip -1em
    \caption{
        Comparison against state-of-the-art BFT DAG-based protocols with 10 replicas and 32\,B commands.
    }
    \label{fig:best_case_10_replica_bft}
\end{figure}

We compare \sysname against state-of-the-art BFT DAG-based consensus protocols to understand how our CFT-optimized design compares to systems designed for BFT.
We deploy 10 replicas with 32\,B commands and compare against two popular BFT DAG protocols: Bullshark and Mysticeti.
Bullshark is a partially synchronous protocol using Narwhal as its DAG layer.
Mysticeti is an uncertified DAG protocol achieving an optimal three-round commitment.
Since both Bullshark and Mysticeti report average latency, we depict the average latency in this evaluation.
\Cref{fig:best_case_10_replica_bft} shows the latency-throughput trade-off.

\para{Throughput}
\sysname and Mysticeti achieve at least 600k\,cmd/sec throughput, as both employ lightweight DAG-based architectures for parallel block processing.
Bullshark attains only 450k\,cmd/sec (25\% lower) due to its heavier Narwhal-based certified DAG design.

\para{Latency advantages}
More importantly, \sysname achieves substantially lower latency than both BFT protocols: 1200\,ms lower than Bullshark and 180\,ms lower latency than Mysticeti.
This latency advantage stems from two key differences.
First, \sysname's typical commit path requires only 2 message delays, fewer than both Mysticeti (3 message delays) and Bullshark (6 message delays).
Second, \sysname targets crash faults and avoids cryptographic operations entirely, while BFT protocols must authenticate every message with signatures and hashes to defend against Byzantine behavior.
The combination of fewer message delays and zero cryptographic overhead enables \sysname to achieve significantly lower latency than BFT alternatives.

%% file: sections/evaluation/cpu.tex
\subsection{Resource utilization}\label{subsec:eval-resources}

\begin{figure}[t]
    \centering
    \includegraphics[width=\columnwidth]{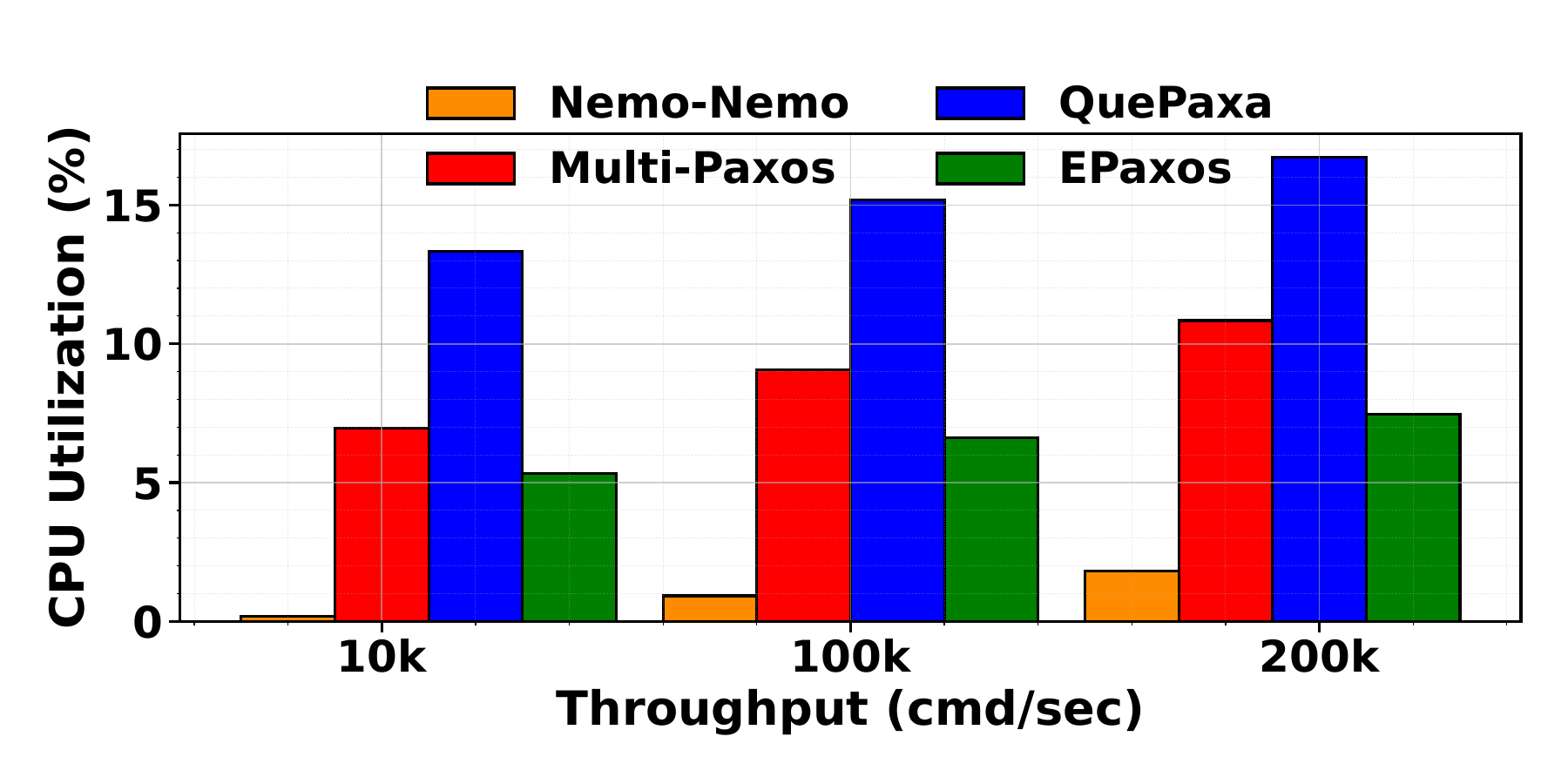}
    \vskip -1em
    \caption{Average CPU utilization across all 5 replicas.}
    \label{fig:cpu_utilization_load}
\end{figure}

We measure the CPU utilization of \sysname and follow Matte et al.~\cite{matte2021scalable}, who use CPU utilization as a proxy for resource consumption.
We record per-replica CPU utilization once per second with Go's ``gopsutil'' tool~\cite{gopsutil}, aggregate the readings across replicas, and report the mean per-replica utilization in \Cref{fig:cpu_utilization_load}.

\Cref{fig:cpu_utilization_load} shows that \sysname uses only $1\%$ CPU at 100k\,cmd/sec, compared to $9\%$, $15\%$, and $7\%$ for Multi-Paxos, QuePaxa, and EPaxos.\footnote{QuePaxa adopts a gRPC-based multi-threaded design, while Multi-Paxos uses a single-threaded one, which increases QuePaxa’s CPU usage.}
The same pattern holds at 10k and 200k\,cmd/sec.
Multi-Paxos, QuePaxa, and EPaxos consume more CPU because they process more messages: \sysname broadcasts one message per block, while existing protocols require at least two propose–vote rounds to commit a batch.
This extra communication directly increases their CPU cost, whereas \sysname avoids it.

%% file: sections/related-work-revision.tex
\section{Related Work}\label{sec:related-work}


\para{Leader-based protocols}
Most deployed consensus protocols rely on a leader to order requests and achieve one-round-trip normal-case commit latency~\cite{lamport2001paxos,oki1988viewstamped,ongaro2014search}.
While simple and efficient in favorable conditions, leader-based protocols suffer from two fundamental limitations: all client requests must funnel through a single leader, limiting throughput scalability, and performance degrades significantly under asynchronous or poor network conditions when the leader or its closest quorum experiences delays.
 \rev{SwiftPaxos~\cite{swiftpaxos} lowers WAN latency by committing a command in two message delays whenever replicas happen to receive commands that touch the same data in the same order. When they disagree on that order, the command falls back to a slower path through the leader. } As shown in~\Cref{fig:best_case_5_replica_throughput_latency,fig:scalability_num_replicas,fig:scalability_message_size,fig:random_case_5_replica_throughput_latency_new_setup}, \sysname outperforms leader-based protocols in both throughput and latency by distributing load evenly across replicas, providing robustness under network randomization.

\para{Multi-leader protocols}
To address the leader bottleneck, several protocols enable multiple concurrent leaders.
Mencius~\cite{barcelona2008mencius} statically partitions the replicated log across replicas, but its main drawback is that SMR progress depends on the slowest replica.
In contrast, \sysname progresses at the speed of the majority.
Generalized Paxos~\cite{lamport2005generalized} and EPaxos~\cite{moraru2013there} support multi-leaders by exploiting request dependencies and allowing partial order of commands.
Under low load, EPaxos achieves lower latency than \sysname, but under high load, the dependency-checking overhead becomes a bottleneck, degrading performance.
Fast Paxos~\cite{lamport2006fast} and Multi-coordinated Paxos~\cite{camargos2007multicoordinated} also permit multiple leaders but deliver suboptimal performance in practice~\cite{moraru2013there}. \rev{ISS~\cite{iss} multiplexes independent instances of a leader-driven protocol over disjoint partitions of the request space, and scales throughput in the WAN to large replica groups. Unlike \sysname, which skips a stalled proposer within the round, ISS recovers from one only at periodic system-wide synchronization points.}

\para{Parallel dissemination}
SADL-RACS~\cite{tennageracs} takes a different approach by decoupling command dissemination from the consensus critical path, following a design similar to Narwhal~\cite{narwhal-tusk}.
Unlike \sysname, SADL-RACS does not build a DAG; instead, it uses the SADL overlay to maintain a per-replica chain of blocks, which are committed when the next RACS block is committed.
This separation requires at least 5 network messages to commit a batch of commands, resulting in high latency overhead (see in~\Cref{fig:best_case_5_replica_throughput_latency}).
In contrast, \sysname embeds command dissemination directly into the DAG, avoiding extra message hops and achieving lower latency and higher throughput.

\rev{\para{Hierarchical protocols}
Steward~\cite{steward}, GeoBFT~\cite{geobft}, and ORION~\cite{orion} scale \emph{Byzantine} consensus to the WAN by grouping replicas per site, reaching agreement within each, and carrying each cluster's decision to the others as a threshold signature or commit certificate. Like \sysname, they avoid funneling all client traffic through one replica. Unlike \sysname, they obtain this by partitioning replicas into per-site clusters, rather than from a single flat set of replicas and a causal DAG.}

\para{DAG-based protocols}
DAG-based consensus protocols are recent advancements in Byzantine fault-tolerance (BFT), but have not been explored in the context of crash fault-tolerant (CFT) systems.
Existing BFT DAG-based protocols can be classified into two categories: certified DAGs~\cite{narwhal-tusk,bullshark} and uncertified DAGs~\cite{cordial-miners,mysticeti2023,mahi-mahi}.
\sysname is the first CFT protocol to adopt a DAG-based architecture, specifically using an uncertified DAG design.

DAG-Rider~\cite{dag-rider}, Tusk~\cite{narwhal-tusk}, Bullshark~\cite{bullshark}, and Dumbo-NG~\cite{dumbo-ng} are representative certified DAG protocols that use consistent broadcast to explicitly certify each DAG block~\cite{sok-dag}.
Fides~\cite{fides} is a TEE-assisted BFT DAG protocol that also uses certified blocks, achieving certification in two rounds without signatures by leveraging TEE guarantees.
Explicit certification ensures that equivocating blocks cannot exist, simplifying the commit rule.
However, certification adds at least 3 message delays per DAG round and increases latency, as seen in Bullshark's results in~\Cref{fig:best_case_10_replica_bft}.
It also raises bandwidth and CPU costs, since replicas must disseminate, receive, and verify cryptographic certificates.
To avoid these drawbacks, \sysname employs an uncertified DAG architecture.

Cordial Miners~\cite{cordial-miners}, Mysticeti~\cite{mysticeti2023}, Mahi-Mahi~\cite{mahi-mahi}, and BlueBottle~\cite{vos2025bluebottle} operate over uncertified DAGs, where each block is disseminated using best-effort to all replicas.
\sysname also builds on an uncertified DAG, but differs fundamentally from these BFT protocols in two key ways.
First, by targeting crash faults instead of Byzantine faults, \sysname's protocol structure requires one less DAG-round to commit, reducing message delays from 3 to 2.
Second, \sysname avoids expensive cryptographic operations such as signature generation and verification, and hash computation and verification.
The combination of fewer message delays and zero cryptographic overhead enables \sysname to achieve significantly lower latency than Mysticeti (see in~\Cref{fig:best_case_10_replica_bft}).

%% file: sections/future-work.tex
\section{Future Work}
\label{sec:FutureWork}

\Rithwik{Think we would have to shorten this a little bit}
\paragraph{The DAG as a responsiveness oracle.}
\label{subsec:oracle}
A property of \sysname of independent interest is that the DAG doubles as a
zero-cost record of peer responsiveness. Because every block references
prior-round blocks and carries its causal history, three signals are available
to every replica by local traversal alone.
First, \emph{liveness}: a replica that keeps producing referenced blocks is
healthy. Second, \emph{super-quorum acknowledgments}: references beyond the
committing $f+1$ reveal how many replicas actually received a block, not merely
that a quorum did. Third, \emph{belated acknowledgments}: slow peers reappear as
references in later rounds, distinguishing slow from crashed. Every node
inherits this estimate for free. As the estimate is local and lagging it guides dissemination only heuristically.

\para{Future Directions}
We leverage this property to motivate an erasure-coded dissemination extension
that reduces bandwidth and storage. One promising direction is to integrate an
eAID-style adaptive dissemination layer~\cite{eAID-Dissemination}. eAID
replaces full-payload broadcast with erasure-coded fragments disseminated
according to the estimated responsiveness of the system: the sender generates
many distinct fragments via Reed--Solomon encoding such that any $f+1$ suffice to
reconstruct the payload, and a node expecting a responsive set of $f+t$ replicas
($f$ the fault threshold, $t$ the responsive nodes beyond it) sends
$\lceil \frac{f+1}{t} \rceil$ fragments per peer, so $f+1$ survive even under
maximum faults. As shown above, the DAG provides accurate,
low-overhead responsiveness estimates
for free. The main tradeoff is encoding/decoding CPU on the payload path, which partially offsets \sysname's CPU advantage; we expect it to be modest relative to the bandwidth savings, but quantifying it is left to future work.

%% file: acks.tex

%% file: sections/proofs-supplementary.tex
\section*{Supplementary Material}

\subsection{Liveness in the random async. model}\label{supp:sec:random-network-liveness}

Under the random asynchronous model~\cite{danezis2025byzantine}, each block includes references to a random set of $f+1$ blocks from the previous round. We now prove that the \sysname protocol achieves liveness with probability $1$ under the random asynchronous model.

\begin{lemma}\label{supp:lem:direct-commit-prob}
    Fix a skeleton slot $S$. A replica directly commits $S$ with probability at least $1/2$.
\end{lemma}
\begin{proof}
    A block in round $R+1$ will reference a random subset of $f+1$ blocks in round $R$. The probability that a given block $B$ votes for the skeleton slot $S$ is $p = \frac{f+1}{2f+1} > \frac{1}{2}$.
    Let $V(S)$ denote the random variable representing the number of replicas in round $R+1$ that vote for $S$ in round $R$. $V(S)$ follows a Binomial distribution:
    $ V(S) \sim B(n, p) \text{, where } p = \frac{f+1}{2f+1}$. The expected value of $V(S)$ is: $E[V(S)] = n \cdot p = f+1$.
    Since the mean $\mu = f+1$ is an integer, the mean and the median of the distribution are equal~\cite{BinomialAverages}. By the definition of the median $m$, we have $\Pr(V(S) \ge m) \ge \frac{1}{2}$. Thus, the probability of a skeleton slot in round $R$ receiving at least $f+1$ votes in round $R+1$, and thus being directly committed, satisfies:
    $\Pr(V(S) \ge f+1) \ge 1/2$.
\end{proof}

\begin{lemma}\label{supp:lem:all-slots-decide-2}
    Fix a skeleton slot $S$. Every correct replica eventually either commits or skips $S$ with probability~$1$.
\end{lemma}
\begin{proof}
    By \Cref{supp:lem:direct-commit-prob}, the probability of a correct replica directly committing any skeleton block in a given round is a constant ($\geq 1/2$).
    Thus, the probability that a sequence of $t$ rounds without any directly committed slot is less than $2^{-t}$. This implies that w.h.p., every slot that is not directly decided will eventually have a committed anchor and therefore become decided (\ie it will be committed or skipped).
\end{proof}

\begin{thm}[Validity]
    \sysname satisfies the validity property of Atomic Broadcast, with probability $1$.
\end{thm}
\begin{proof}
    Let $N$ be a correct replica and $B$ a block proposed by $N$. We show that, with probability $1$, $B$ is eventually committed by every correct replica. By \Cref{lem:eventually-all-references}, $B$ is eventually included in the local DAG of every correct replica.
    So every correct replica will eventually include a reference to $B$ in at least one of its blocks. Let $R$ be the highest round at which some correct replica includes a reference to $B$ in one of its blocks. By \Cref{supp:lem:direct-commit-prob}, with probability $1$, each skeleton block has a nonzero probability of being directly committed, so eventually some block $B'$ at a round $R' > R$ will be directly committed. Since all replicas have $B$ in their causal histories by round $R$, $B'$ must therefore have a path to $B$. \Cref{supp:lem:all-slots-decide-2} guarantees that all slots before $B'$ are eventually decided, so $B'$ is eventually committed. Thus, $B$ will be committed at all correct replicas at the latest when $B'$ is committed along with its causal history.
\end{proof}

\begin{thm}[Agreement]
    \sysname satisfies the agreement property of Atomic Broadcast, with probability $1$.
\end{thm}
\begin{proof}
    Let $N$ be a correct replica and $B$ a block committed by $N$. We show that, with probability $1$, $B$ is eventually committed by every correct replica. Let $L$ be the skeleton block with which $B$ is committed and $S$ the corresponding slot. By \Cref{supp:lem:all-slots-decide-2}, all blocks up to and including $S$ are eventually decided by all correct replicas, with probability $1$. By \Cref{obs:agree-commit}, all correct replicas commit $L$ in $S$. Eventually, all correct replicas commit $B$.
\end{proof}